\pdfoutput=1
\documentclass[12pt,fleqn]{article}

\usepackage{amsmath,amssymb,amsthm,mathtools,thm-restate}
\usepackage{fullpage}
\usepackage{indentfirst}

\usepackage{comment}

\usepackage[T1]{fontenc}
\usepackage[scale=1.05]{biolinum}  
\usepackage{fouriernc}
\usepackage[cal=pxtx,scr=boondox,bb=stixtwo,frak=pxtx]{mathalpha}
\let\newmathbb\mathbb
\AtBeginDocument{%
    \let\mathbb\relax
    \newcommand{\mathbb}[1]{\newmathbb{#1}}
}
\renewcommand{\mathsf}[1]{\text{\upshape\sffamily#1}}

\usepackage[style=alphabetic,natbib=true,sortcites=true,backref=true,maxnames=99,maxalphanames=99]{biblatex}

\DefineBibliographyStrings{english}{%
  backrefpage = {$\Lsh$~p.},
  backrefpages = {$\Lsh$~pp.},
}
\usepackage[colorlinks=true,citecolor=blue!65!black,linkcolor=red!75!black]{hyperref}
\usepackage[nameinlink]{cleveref}

\usepackage[basic,sanserif]{complexity}

\usepackage{booktabs,colortbl}
\usepackage{diagbox,enumitem}
\usepackage{framed,fancybox,ascmac}

\usepackage{comment}
\usepackage{url}

\usepackage[colorinlistoftodos,prependcaption]{todonotes}

\usepackage{xspace}
\usepackage{bm}

\usepackage{lineno}

\usepackage{tikz}
\usetikzlibrary{positioning}
\usetikzlibrary{decorations.markings}
\usetikzlibrary{shapes.geometric}
\usetikzlibrary{arrows.meta}
\usetikzlibrary{arrows,automata,decorations.pathmorphing,fit,positioning,arrows.meta,calc}

\crefname{theorem}{Theorem}{Theorems}
\crefname{proposition}{Proposition}{Propositions}
\crefname{lemma}{Lemma}{Lemmas}
\crefname{claim}{Claim}{Claims}
\crefname{corollary}{Corollary}{Corollaries}
\crefname{remark}{Remark}{Remarks}
\crefname{observation}{Observation}{Observations}
\crefname{hypothesis}{Hypothesis}{Hypotheses}
\crefname{definition}{Definition}{Definitions}
\crefname{problem}{Problem}{Problems}
\crefname{example}{Example}{Examples}

\crefname{appendix}{Appendix}{Appendices}
\crefname{section}{Section}{Sections}
\crefname{equation}{Eq.}{Eqs.}
\crefname{figure}{Figure}{Figures}
\crefname{table}{Table}{Tables}

\crefname{algorithm}{Algorithm}{Algorithms}

\usepackage[ruled]{algorithm}
\usepackage{algpseudocodex}
\algnewcommand{\algorithmicand}{\textbf{ and }}
\algnewcommand{\algorithmicor}{\textbf{ or }}
\algnewcommand{\algorithmicto}{\textbf{ to }}

\algrenewcommand\textproc{\textsl}

\renewcommand{\geq}{\geqslant}
\renewcommand{\leq}{\leqslant}
\renewcommand{\phi}{\varphi}
\renewcommand{\epsilon}{\varepsilon}
\renewcommand{\bar}{\overline}

\renewcommand{\tilde}{\widetilde}

\newcommand{\prb}[1]{\textup{\textsc{#1}}\xspace}

\renewcommand{\vec}[1]{\mathbf{\bm{#1}}}

\newcommand{\reco}{\leftrightsquigarrow}

\DeclareMathOperator{\val}{\mathsf{val}}

\DeclareMathOperator{\bigO}{\mathcal{O}}

\DeclareMathOperator*{\argmin}{argmin}
\let\Pr\relax\DeclareMathOperator*{\Pr}{\mathbb{Pr}}
\DeclareMathOperator{\Enc}{\mathsf{Enc}}
\DeclareMathOperator{\Dec}{\mathsf{Dec}}

\newcommand{\sss}{\mathsf{start}}
\newcommand{\ttt}{\mathsf{goal}}

\newcommand{\ckt}{\mathrm{ckt}}

\newcommand{\asgmt}{A}                 
\newcommand{\sqasgmt}{\scrA}          
\newcommand{\sqsigma}{\sigma}
\newcommand{\sqgamma}{\gamma}
\newcommand{\sqpi}{\pi}

\newcommand{\bbN}{\mathbb{N}}

\newcommand{\scrA}{\mathscr{A}}

\newcommand{\scrC}{\mathscr{C}}
\newcommand{\scrD}{\mathscr{D}}
\newcommand{\scrE}{\mathscr{E}}

\newtheorem{theorem}{Theorem}[section]
\newtheorem{proposition}[theorem]{Proposition}
\newtheorem{lemma}[theorem]{Lemma}
\newtheorem{claim}[theorem]{Claim}
\newtheorem{corollary}[theorem]{Corollary}
\newtheorem{observation}[theorem]{Observation}

\theoremstyle{definition}
\newtheorem{definition}[theorem]{Definition}

\newtheorem{problem}[theorem]{Problem}
\newtheorem{example}[theorem]{Example}
\newenvironment{claim*}{\begin{claim}}{\end{claim}}

\numberwithin{equation}{section}

\title{Probabilistically Checkable Reconfiguration Proofs and Inapproximability of Reconfiguration Problems}
\author{
Shuichi Hirahara\thanks{National Institute of Informatics, Japan.
\href{mailto:s\_hirahara@nii.ac.jp}{\texttt{s\_hirahara@nii.ac.jp}}
}
\and
Naoto Ohsaka\thanks{CyberAgent, Inc., Tokyo, Japan. \href{mailto:ohsaka\_naoto@cyberagent.co.jp}{\texttt{ohsaka\_naoto@cyberagent.co.jp}}; \href{mailto:naoto.ohsaka@gmail.com}{\texttt{naoto.ohsaka@gmail.com}}
}
}
\date{\today}

\begin{document}

\maketitle

\begin{abstract}Motivated by the inapproximability of reconfiguration problems,
we present a new PCP-type characterization of $\PSPACE$, which we call 
a \emph{probabilistically checkable reconfiguration proof} (PCRP):
Any $\PSPACE$ computation can be encoded into an exponentially long sequence of polynomially long proofs
such that 
every adjacent pair of the proofs differs in at most one bit,
and every proof can be probabilistically checked by reading a constant number of bits.

Using the new characterization,  we prove $\PSPACE$-completeness of approximate versions of many reconfiguration problems,
such as the \prb{Maxmin $3$-SAT Reconfiguration} problem.
This resolves the open problem posed by {Ito, Demaine, Harvey, Papadimitriou, Sideri, Uehara, and Uno}~(ISAAC 2008; Theor.~Comput.~Sci.~2011)
as well as the Reconfiguration Inapproximability Hypothesis by Ohsaka (STACS 2023) affirmatively.
We also present $\PSPACE$-completeness of approximating the \prb{Maxmin Clique Reconfiguration} problem
to within a factor of $n^\epsilon$ for some constant $\epsilon > 0$.

\end{abstract}

\clearpage

\tableofcontents

\clearpage

\section{Introduction}

\emph{Reconfiguration problems} ask to decide whether there exists a sequence of operations that 
transform one feasible solution to another.  A canonical example is the \prb{3-SAT Reconfiguration} problem,
which is known to be $\PSPACE$-complete \cite{gopalan2009connectivity}.
\begin{definition}
    [\prb{$3$-SAT Reconfiguration} \cite{gopalan2009connectivity}]
Given a $3$-CNF formula $\phi$ and 
its two satisfying assignments $\sigma^\sss$ and $\sigma^\ttt$,
we are required to decide if
there is a sequence of satisfying assignments to $\phi$,
$(\sigma^{(1)}, \ldots, \sigma^{(T)})$,
such that
$\sigma^{(1)} = \sigma^\sss$,
$\sigma^{(T)} = \sigma^\ttt$, and
$\sigma^{(t)}$ and $\sigma^{(t+1)}$ differ in at most one variable for every
$t \in \{1, \ldots, T-1\}$.
\end{definition}

\begin{example}
Suppose we are given a $3$-CNF formula
$\phi \coloneq (\bar{x_1} \vee \bar{x_2} \vee x_3) \wedge (\bar{x_1} \vee x_2 \vee \bar{x_3}) \wedge (x_1 \vee \bar{x_2} \vee \bar{x_3})$
made up of three clauses over three variables $x_1,x_2,x_3$ and
its two satisfying assignments
$\sigma_1^\sss \coloneq (1,0,0)$ and
$\sigma_1^\ttt \coloneq (0,1,0)$.
Then,
$(\phi, \sigma_1^\sss, \sigma_1^\ttt)$ is a YES instance of \prb{$3$-SAT Reconfiguration}:
there exists a sequence
$((1,0,0), (0,0,0), (0,1,0))$ from $\sigma_1^\sss$ to $\sigma_1^\ttt$ that
meets the requirement.
On the other hand,
if we are given a pair of two satisfying assignments
$\sigma_2^\sss \coloneq (1,0,0)$ and
$\sigma_2^\ttt \coloneq (1,1,1)$,
then $(\phi, \sigma_2^\sss, \sigma_2^\ttt)$ is a NO instance because
any sequence from $\sigma_2^\sss$ to $\sigma_2^\ttt$ must run through
$(0,1,1)$, $(1,0,1)$, or $(1,1,0)$,
neither of which satisfy $\phi$.
\end{example}

It is natural to consider its approximate variant, whose complexity was posed as an open problem in \cite{ito2011complexity}. 

\begin{definition}
    [\prb{Maxmin $3$-SAT Reconfiguration} \cite{ito2011complexity}]
Given a $3$-CNF formula $\phi$ over $m$ clauses $C_1, \ldots, C_m$ and
its two satisfying assignments $\sigma^\sss$ and $\sigma^\ttt$,
we are required to find a sequence of assignments,
$(\sigma^{(1)}, \ldots, \sigma^{(T)})$,
such that
$\sigma^{(1)} = \sigma^\sss$,
$\sigma^{(T)} = \sigma^\ttt$,
$\sigma^{(t)}$ and $\sigma^{(t+1)}$ differ in at most one variable for every
$t \in \{1, \ldots, T-1\}$, and
the following objective value is maximized:
\begin{align}
    \min_{1 \leq t \leq T}
    \frac{1}{m} \left|\Bigl\{
        j \in \{1, \ldots, m\} \Bigm| \sigma^{(t)} \text{ satisfies } C_j
    \Bigr\}\right|.
\end{align}
\end{definition}
\begin{example}
Consider again the same $3$-CNF formula $\phi$ and its two satisfying assignments
$\sigma_2^\sss = (1,0,0)$ and $\sigma_2^\ttt = (1,1,1)$.
There is a sequence
$((1,0,0), (1,1,0), (1,1,1))$
from $\sigma_2^\sss$ to $\sigma_2^\ttt$,
which is a feasible solution to \prb{Maxmin $3$-SAT Reconfiguration} and
whose objective value is $\frac{2}{3}$.
\end{example}

The main contribution of this paper is to prove $\PSPACE$-completeness of approximating the \prb{Maxmin $3$-SAT Reconfiguration} problem within a constant factor, which answers the open problem of \cite{ito2011complexity}.
In what follows, we present the background of this result and then the details of our results.

\subsection{Background}
Given a \emph{source problem} that asks the existence of a feasible solution,
\emph{reconfiguration problems} are defined as a problem
of deciding the existence of 
a \emph{reconfiguration sequence}, that is, 
a step-by-step transformation between a pair of feasible solutions while always preserving the feasibility of solutions.
For example,
\prb{$3$-SAT Reconfiguration} \cite{gopalan2009connectivity}
is defined from \prb{$3$-SAT} as a source problem.
Many reconfiguration problems can be defined from
Boolean satisfiability, constraint satisfaction problems, graph problems, and others.
Studying reconfiguration problems
may help elucidate the structure of the solution space \cite{gopalan2009connectivity},
which is motivated by, e.g.,
the application to the behavior analysis of SAT solvers,
such as DPLL \cite{achlioptas2004exponential}.
From a different point of view,
reconfiguration problems may date back to motion planning \cite{hopcroft1984complexity} and
classical puzzles, including 15 puzzles \cite{johnson1879notes} and Rubik's Cube.

Typically,
a reconfiguration problem becomes $\PSPACE$-complete if its source problem is intractable (say, $\NP$-complete);
e.g., \prb{$3$-SAT} \cite{gopalan2009connectivity},
\prb{Independent Set} \cite{hearn2005pspace,hearn2009games},
\prb{Set Cover} \cite{ito2011complexity}, and
\prb{$4$-Coloring} \cite{bonsma2009finding}.
On the other hand,
a source problem in $\P$ frequently leads to a reconfiguration problem in $\P$,
e.g.,
\prb{Matching} \cite{ito2011complexity} and
\prb{$2$-SAT} \cite{gopalan2009connectivity}.
Some exceptions are known:
whereas \prb{$3$-Coloring} is $\NP$-complete,
its reconfiguration problem is solvable in polynomial time \cite{cereceda2011finding};
\prb{Shortest Path} on a graph is tractable,
but its reconfiguration problem is $\PSPACE$-complete \cite{bonsma2013complexity}.
We refer the readers to the surveys by
\citet{nishimura2018introduction,heuvel13complexity} for algorithmic and hardness results and
the Combinatorial Reconfiguration wiki \cite{hoang2023combinatorial} for an exhaustive list of related articles.

A common way to cope with intractable problems is to consider approximation problems.
Relaxing the feasibility of intermediate solutions,
we can formalize \emph{approximate variants} for reconfiguration problems,
which are also motivated by the situation wherein
there does not exist a reconfiguration sequence for the original decision problem.
For example, in \prb{Maxmin $3$-SAT Reconfiguration} \cite{ito2011complexity},
we are allowed to include any \emph{non-satisfying} assignment in a reconfiguration sequence, but 
required to maximize the \emph{minimum fraction of satisfied clauses}.
Solving this problem may result in a reasonable reconfiguration sequence
consisting of \emph{almost-satisfying} assignments, e.g.,
each violating at most $1\%$ of clauses.
Intriguingly, a different trend regarding the approximability has been observed
between a source problem and its reconfiguration analogue; e.g.,
\prb{Set Cover} is $\NP$-hard to approximate within a factor better than $\ln n$
\cite{lund1994hardness,feige1998threshold,dinur2014analytical}, whereas
\prb{Minmax Set Cover Reconfiguration} admits a $2$-factor approximation algorithm \cite{ito2011complexity}.
Other reconfiguration problems whose approximability was investigated include:
\prb{Subset Sum Reconfiguration} has a PTAS \cite{ito2014approximability};
\prb{Submodular Reconfiguration} \cite{ohsaka2022reconfiguration} and
\prb{Power Supply Reconfiguration} \cite{ito2011complexity} are
constant-factor approximable.

Little is known about the hardness of approximation for reconfiguration problems.
Using the fact that source problems (e.g., \prb{Max $3$-SAT}) are $\NP$-hard to approximate \cite{arora1998probabilistic,arora1998proof},
\citet{ito2011complexity} proved that 
several reconfiguration problems (e.g., \prb{Maxmin $3$-SAT Reconfiguration}) are $\NP$-hard to approximate;
however, most reconfiguration problems are in $\PSPACE$, and thus their $\NP$-hardness results are not optimal.
It was left open to improve the $\NP$-hardness results to $\PSPACE$-hardness.
We here stress the significance of showing $\PSPACE$-hardness compared to $\NP$-hardness:
\begin{enumerate}
    \item $\PSPACE$-hardness is \emph{tight} because most reconfiguration problems belong to $\PSPACE$ \cite{nishimura2018introduction};
    \item it disproves the existence of a polynomial-length witness (in particular, a polynomial-length reconfiguration sequence)
    assuming $\NP \neq \PSPACE$;
    \item it rules out any polynomial-time algorithm under the weak assumption that $\P \neq \PSPACE$.
\end{enumerate}

In order to improve the $\NP$-hardness of approximation to $\PSPACE$-hardness of approximation,
it is crucial to develop a reconfiguration analogue of the PCP theorem \cite{arora1998probabilistic,arora1998proof}.
%
As indicated by the gap in the approximation factors of \prb{Set Cover} and its reconfiguration counterpart,
the required theory must be different and tailored to $\PSPACE$.
\citet{ohsaka2023gap} recently postulated a reconfiguration analogue of the PCP theorem, called the \emph{Reconfiguration Inapproximability Hypothesis} (RIH), under which 
a bunch of popular reconfiguration problems are shown to be $\PSPACE$-hard to approximate.
The major open question is whether RIH holds.

\subsection{Our Results}
Our contribution is to 
present a new PCP-type characterization of $\PSPACE$, which we call a
\emph{probabilistically checkable reconfiguration proof} (PCPR), and 
thereby affirmatively resolve the open problem posed by \citet{ito2011complexity}
and confirm RIH of \citet{ohsaka2023gap}.

Our characterization of $\PSPACE$ encodes
any $\PSPACE$ computation into
an exponentially long reconfiguration sequence of polynomial-length proofs,
each of which can be probabilistically checked by reading a constant number of bits.
A \emph{reconfiguration sequence from $\pi^\sss$ to $\pi^\ttt$ over $\{0, 1\}^n$} is a sequence $( \pi^{(1)}, \cdots, \pi^{(T)} ) \in (\{0,1\}^n)^*$
such that 
$\pi^\sss = \pi^{(1)}$, $\pi^\ttt = \pi^{(T)}$, and
$\pi^{(t)}$ and $\pi^{(t+1)}$ differ in at most one bit for every $t \in \{1, \cdots, T - 1\}$.

\begin{theorem}[Probabilistically Checkable Reconfiguration Proof (PCRP); see also \cref{thm:reconfPCP}]
\label{thm:reconfPCP in intro}
A language $L$ is in $\PSPACE$ if and only if 
there exist a randomized polynomial-time verifier $V$
with randomness complexity $\bigO(\log n)$ and query complexity $\bigO(1)$
on inputs of length $n$
and polynomial-time algorithms $\pi^\sss$ and $\pi^\ttt$
with the following properties:
\begin{enumerate}
    \item \textup{(Completeness)} If $x \in L$, then there exists a reconfiguration sequence $( \pi^{(1)}, \cdots, \pi^{(T)} )$ from $\pi^\sss(x)$ to $\pi^\ttt(x)$ over $\{0,1\}^{\poly(n)}$ such that 
    for every $t \in \{1, \cdots, T\}$,
    \[
      \Pr \Bigl[ V^{\pi^{(t)}}(x) = 1 \Bigr] = 1.
    \]
    \item \textup{(Soundness)} If $x \not\in L$, then for every reconfiguration sequence $( \pi^{(1)}, \cdots, \pi^{(T)} )$
    from $\pi^\sss(x)$ to $\pi^\ttt(x)$,
    for some $t \in \{1, \cdots, T \}$,
    \[
      \Pr \Bigl[ V^{\pi^{(t)}}(x) = 1 \Bigr] < \frac{1}{2}.
    \]
\end{enumerate}
Here, 
$V^{\pi^{(t)}}(x)$ denotes the output of $V$ on input $x$ given oracle access to $\pi^{(t)}$,
and  the probabilities are over the $\bigO(\log n)$ random bits of the verifier $V$.
\end{theorem}

The verifier $V$ can be regarded as a $\forall \cdot \co\RP$-type verifier:
The verifier co-nondeterministically guesses $t \in \{1, \cdots, T\}$
and probabilistically checks the $t$-th proof $\pi^{(t)}$.
This verifier should be compared with the standard $\co\RP$-type PCP verifier $V'$ for $\PSPACE$-complete problems, which can be obtained from the PCP theorem for $\NEXP \supseteq \PSPACE$ \cite{BabaiFL91_cc_journals}.
The number of random bits used by $V'$ is $n^{\Theta(1)},$ 
whereas the number of random bits of our verifier $V$ is $\bigO(\log n)$.
The latter is crucial for the application to inapproximability of reconfiguration problems.
The standard verifier $V'$ uses only random bits, whereas our verifier $V$ co-nondeterministically guesses $t$.
Given that $V'$ does not use any nondeterministic choice,
it is natural to wonder whether $V$ can be improved to a $\co\RP$-type verifier that chooses $t \in \{1, \cdots, T\}$ randomly;
however, such an extension is impossible (see  \cref{obs:impossible}), and thus our characterization is one of the ``best'' characterizations in this direction.

%

As a corollary of \cref{thm:reconfPCP} and \cite{ohsaka2023gap,ohsaka2024gap},
we obtain that a host of reconfiguration problems are $\PSPACE$-complete to approximate. 
\begin{corollary}[from \cref{thm:reconfPCP} and \cite{ohsaka2023gap,ohsaka2024gap}]
\label{cor:RIH}
For some universal constant $\epsilon_0 \in (0,1)$,
the following approximate variants of reconfiguration problems are 
$\PSPACE$-hard to approximate within a factor of $(1-\epsilon_0)$.
\begin{itemize}
\item \prb{Maxmin $k$-SAT Reconfiguration} for all $k \geq 2$;
\item \prb{Maxmin $q$-CSP Reconfiguration} for all $q \geq 2$;
\item \prb{Maxmin Independent Set Reconfiguration} on bounded-degree graphs;
\item \prb{Minmax Vertex Cover Reconfiguration} on bounded-degree graphs;
\item \prb{Maxmin Clique Reconfiguration};
\item \prb{Minmax Dominating Set Reconfiguration};
\item \prb{Minmax Set Cover Reconfiguration};
\item \prb{Maxmin Nondeterministic Constraint Logic}.
\end{itemize}
\end{corollary}

Moreover, we improve an inapproximability factor of \prb{Maxmin Clique Reconfiguration} to a polynomial; that is,
\prb{Maxmin Clique Reconfiguration} is $\PSPACE$-hard to approximate
within a factor of $n^\epsilon$ for some constant $\epsilon > 0$,
where $n$ is the number of vertices
(\cref{thm:clique}).
This is the first polynomial-factor inapproximability result for
approximate variants of reconfiguration problems (to the best of our knowledge).\footnote{
See \cref{sec:related} for discussion of existing hardness-of-approximation results
for reconfiguration problems.
}

\section{Proof Overview}

Here, we present a proof sketch of  \cref{thm:reconfPCP in intro}.

A na\"ive attempt for the proof of \cref{thm:reconfPCP in intro}
would be to develop reconfiguration counterparts for 
the simple proof of the PCP theorem by \citet{dinur2007pcp}.
The proof of the PCP theorem consists of repeated applications of the three steps --- a degree reduction (the preprocessing lemma \cite[Lemma 1.9]{dinur2007pcp}), gap amplification and an alphabet reduction.
Counterparts of some of the steps have been developed in the recent literature of reconfiguration problems
\cite{ohsaka2023gap,ohsaka2024gap}.
For example, \citet{ohsaka2023gap} presented a degree reduction for reconfiguration problems, i.e., 
a reduction that converts a graph that represents a PCRP (probabilistically checkable reconfiguration proof) system 
with soundness error $1 - \epsilon$
into another graph whose degree $\Delta$ is small.
However, the parameter achieved in \cite{ohsaka2023gap} is weaker than that of \cite{dinur2007pcp,papadimitriou1991optimization}:
$\Delta \le \poly(1 / \epsilon)$.
If $\epsilon = o(1)$, the degree $\Delta$ can be $\omega(1)$, which is not sufficient for Dinur's proof to go through.
It appears to be very difficult to construct PCRPs
based on this approach.

Our actual approach is much simpler.  We use existing machinery developed in the literature of PCP theorems in a black-box way.
The main ingredient for our proof is the \emph{PCP of Proximity} (PCPP) \cite{dinur2006assignment,bensasson2006robust}.  
A PCPP for a language $L \in \NP$ allows us to approximately verify that $x \in L$
by reading a constant number of bits from $x$ and a proof.
In particular, by encoding $x$ by an error-correcting code,
we can reliably check whether $x \in L$ efficiently.

To construct a PCRP for every problem in $\PSPACE$, it suffices to construct a PCRP for some $\PSPACE$-complete problem.
We consider the $\PSPACE$-complete problem called  \prb{Succinct Graph Reachability}.
In what follows, we first explain how this problem can be regarded as a reconfiguration problem,
and then explain how to construct a PCRP system for \prb{Succinct Graph Reachability}.

\subsection{\prb{Succinct Graph Reachability} as Reconfiguration Problems}

\prb{Succinct Graph Reachability} is the following problem.
The input consists of a circuit which succinctly represents an exponentially large graph $G = (V, E)$
and two vertices $v^\sss$ and $v^\ttt \in V$,
and the task is to decide whether there exists a path from $v^\sss$ to $v^\ttt$ in $G$.
Each vertex is represented by an $n$-bit string; i.e., $V = \{0,1\}^n$.
For simplicity of notation, throughout this section, we assume that every vertex in $G$ has a self-loop, i.e., $(x, x) \in E$ for every $x \in V$.
For two strings $x$ and $y$, we denote by $x \circ y$ the concatenation of $x$ and $y$.

\prb{Succinct Graph Reachability} can be naturally regarded as the following reconfiguration problem.
Given a (succinctly described) graph $G$ and two vertices $v^\sss, v^\ttt \in V$,
the task is to decide whether 
there exists a sequence $(x_1 \circ y_1, \cdots, x_T \circ y_T) \in (\{0, 1\}^{2 n})^*$ from $v^\sss \circ v^\sss$ to $v^\ttt \circ v^\ttt$
such that 
\begin{enumerate}
    \item every configuration $x_t \circ y_t \in \{0, 1\}^{2 n}$ satisfies the constraint that $(x_t, y_t) \in E$, and 
    \item each adjacent pair of configurations satisfy $x_t = x_{t+1}$ or $y_t = y_{t+1}$.
\end{enumerate}
In other words, this is the reconfiguration problem which asks to decide whether
the token that initially placed at the edge $(v^\sss, v^\sss)$
can be moved to the edge $(v^\ttt, v^\ttt)$
by a sequence of operations that move the token from an edge to one of its adjacent edges.

In \cref{thm:reconfPCP in intro}, each adjacent pair of proofs differs in at most one bit.  In terms of reconfiguration problems, this means that the operations which we are allowed to perform are to change \emph{one bit} of a configuration instead of \emph{one vertex} of the token placed at an edge.
By introducing a special symbol ``$\bot$'', we can regard \prb{Succinct Graph Reachability} as the following reconfiguration problem
in which operations are restricted to changing one bit of configurations:
Given a (succinctly described) graph $G$ and two vertices $v^\sss, v^\ttt \in V$,
the task is to decide whether 
there exists a sequence $(x_1 \circ y_1, \cdots, x_T \circ y_T) \in (\{0, 1, \bot\}^{2 n})^*$ from $v^\sss \circ v^\sss$ to $v^\ttt \circ v^\ttt$
such that 
\begin{enumerate}
    \item 
    \label{item:reconfiguration constraint}
    $(x_t, y_t) \in E$ or ($x_t \in \{0,1\}^n$ and $y_t \in \{0,1\}^n$), and
    \item each adjacent pair $(x_t \circ y_t, x_{t+1} \circ y_{t+1})$ differs in at most one position.
\end{enumerate}
Informally, the existence of the symbol $\bot$ indicates that we are on the way of the transition,
and we are allowed to include $\bot$ in at most one of $x_t$ or $y_t$
(that is, we do not allow to change both vertices of the token simultaneously).
This reconfiguration problem ``simulates'' \prb{Succinct Graph Reachability} in the following sense:
If a token placed at an edge $(x, y_1) \in E$ is moved to another edge $(x, y_2) \in E$,
then
in the new reconfiguration problem,
we may consider a sequence of operations that first transform $x \circ y_1$ into $x \circ \bot^n$
by replacing each bit of $y_1$ with $\bot$ one by one,
and then transform $x \circ \bot^n$ into $x \circ y_2$
by replacing $\bot$ with a bit of $y_2$ one by one.



\subsection{PCRP System for \prb{Succinct Graph Reachability}}

The main idea for constructing a PCRP for \prb{Succinct Graph Reachability}
is to probabilistically check  \cref{item:reconfiguration constraint}, i.e., the condition that 
$(x_t, y_t) \in E$ or ($x_t \in \{0,1\}^n$ and $y_t \in \{0,1\}^n$),
by reading a constant number of bits.
To this end, we encode each vertex by a locally testable error-correcting code $\Enc \colon \{0,1\}^n \to \{0,1\}^\ell$
and use the PCPP to check whether the encoded pair of vertices $(x_t, y_t)$ satisfies $(x_t, y_t) \in E$.
Specifically, let $V_{\mathrm{PCPP}}$ be a PCPP verifier for the language $L_G = \{ \Enc(x) \circ \Enc(y) \mid (x, y) \in E \}$.
This verifier takes random access to $f \circ g$ and a proof $\pi \in \{0,1\}^p$ and checks whether $f \circ g$ is close to some $\Enc(x) \circ \Enc(y) \in L_G$.
Then, 
given a sequence of proofs $(\sigma^{(1)}, \cdots, \sigma^{(T)}) \in (\{0, 1, \bot\}^{2 \ell + p})^*$,
we probabilistically check each proof $\sigma^{(t)} = f \circ g \circ \pi$ as follows:
\begin{enumerate}
    \item  Using the local tester for $\Enc$, we check that both $f$ and $g$ are close to some codewords $\Enc(x)$ and $\Enc(y)$, respectively.  If both are far from codewords, then we reject.
    \item By random sampling, we test whether either $f$ or $g$ contains many $\bot$ symbols.  If so, we accept.  
    \item Finally, by running $V_\mathrm{PCPP}$ for $(f \circ g, \pi)$, we check that $(x, y) \in E$.
    We accept if and only if $V_\mathrm{PCPP}$ accepts.  
\end{enumerate}
The first item ensures that 
either $f$ or $g$ is close to some codewords $\Enc(x)$ and $\Enc(y)$.
The second item checks whether we are on the way of the transition from one edge to another, in which case we accept.
We run the test of the third item only if either $f$ or $g$ is close to some codewords, and both $f$ and $g$ do not contain many $\bot$ symbols.
Using the PCPP, we check that $f$ and $g$ encode $x$ and $y$ such that $(x, y) \in E$.

We note that the size of alphabets $\{0, 1, \bot\}$ of the PCRP system is $3$.
This can be reduced to $2$ by using a simple alphabet reduction of \citet{ohsaka2023gap},
which transforms any PCRP system with \emph{perfect completeness} over alphabets of constant size 
into a PCRP system over the binary alphabets $\{0, 1\}$.

It is thus important to make sure that the PCRP system has perfect completeness, i.e., in the YES case, the verifier accepts with probability 1.
For this reason, in the actual proof, we need to modify the PCPP verifier $V_\mathrm{PCPP}$ so that it immediately accepts if a $\bot$ symbol in $f$ or $g$ is queried by $V_\mathrm{PCPP}$.
Details can be found in \cref{sec:PCRP}.

%
%

\section{Related Work}
\label{sec:related}

Another characterization of $\PSPACE$ is
\emph{probabilistically checkable debate systems} due to
\citet{condon1995probabilistically},
which can be used to show \PSPACE-hardness of approximating
\prb{Quantified Boolean Formula} and the problem of selecting as many finite-state automata as possible that accept a common string.
These results are incomparable to our PCRP
because the underlying structure of the problems is different from each other.

We summarize approximate variants of reconfiguration problems whose inapproximability was investigated.
\prb{Maxmin Clique Reconfiguration} and \prb{Maxmin SAT Reconfiguration}
are $\NP$-hard to approximate~\cite{ito2011complexity}.
\prb{Shortest Path Reconfiguration} is
$\PSPACE$-hard to approximate with respect to its objective value \cite{gajjar2022reconfiguring}.
The obejctive value called the \emph{price}
is determined based on the number of vertices in a path changed at a time,
which is fundamentally different from those of reconfiguration problems listed in \cref{cor:RIH}.
\prb{Submodular Reconfiguration} is constant-factor inapproximable \cite{ohsaka2022reconfiguration},
whose proof resorts to inapproximability results of \prb{Submodular Function Maximization} \cite{feige2011maximizing}.

We note that approximability of reconfiguration problems frequently refers to
that of \emph{the shortest sequence}
\cite{kaminski2011shortest,miltzow2016approximation,bonamy2020shortest,ito2022shortest},
which seems orthogonal to the present study.

The \emph{pebble game} \cite{paterson1970comparative} is a single-player game,
which models the trade-off between the memory usage and running time of a computation and
is recently used in the context of proof complexity \cite{nordstrom2013pebble}.
This game can be thought of as a reconfiguration problem,
whose objective function, called the \emph{pebbling price},
is defined as the maximum number of pebbles at any time required to place a pebble to the unique sink.
The pebbling price is known to be $\PSPACE$-hard to approximate within an additive $n^{\frac{1}{3}-\epsilon}$ term for the graph size $n$ \cite{chan2015hardness,demaine2017inapproximability}.
We leave open whether our PCRPs can be used to derive
$\PSPACE$-hardness of approximating the pebbling price within a multiplicative factor.

\section{Preliminaries}

\subsection{Notations}
For a nonnegative integer $n \in \bbN$,
let $ [n] \coloneq \{1, 2, \ldots, n\} $.
A \emph{sequence} $\scrE$ of a finite number of elements $E^{(1)}, \ldots, E^{(T)}$
is denoted by $( E^{(1)}, \ldots, E^{(T)} )$, and
we write $E \in \scrE$ to indicate that $E$ appears in $\scrE$.
The symbol $\circ$ stands for a concatenation of two strings.
Let $\Sigma$ be a finite set called \emph{alphabet}.
For a length-$n$ string $f \in \Sigma^n$ and index set $I \subseteq [n]$,
we use $f|_I$ to denote the restriction of $f$ to $I$.
We write $0^n$ and $1^n$ for
$\underbrace{0 \cdots 0}_{n \text{ times}}$ and
$\underbrace{1 \cdots 1}_{n \text{ times}}$,
respectively.
The \emph{relative distance} between two strings $f,g \in \Sigma^n$,
denoted $\Delta(f,g)$,
is defined as the fraction of positions on which $f$ and $g$ differ; namely,
$
    \Delta(f,g)
    = \Pr_{i \sim [n]}[f_i \neq g_i ]
    = \frac{|\{i \in [n] \mid f_i \neq g_i \}|}{n}.
$
We say that
$f$ is \emph{$\epsilon$-close} to $g$ if 
$\Delta(f,g) \leq \epsilon$
and \emph{$\epsilon$-far} from $g$ if
$\Delta(f,g) > \epsilon$.
For a set of strings $S \subseteq \Sigma^n$, analogous notions are defined; e.g.,
$\Delta(f, S) \coloneq \min_{g \in S} \Delta(f,g)$ and
$f$ is $\epsilon$-close to $S$ if $\Delta(f,S) \leq \epsilon$.

\subsection{Error-Correcting and Locally Testable Codes}
Here, we introduce error-correcting and locally testable codes.

\begin{definition}[Error-correcting codes]
For any $\rho \in [0,1]$,
a function $\Enc \colon \{0,1\}^n \to \{0,1\}^\ell$
is an \emph{error-correcting code with relative distance} $\rho$
if
$\Delta(\Enc(\alpha), \Enc(\beta)) > \rho$ for every $\alpha \neq \beta \in \{0,1\}^n$.
We call $\Enc(\alpha)$ for each $\alpha \in \{0,1\}^n$ a \emph{codeword} of $\Enc$.
Denote by $\Enc(\cdot)$ the set of all codewords of $\Enc$.
\end{definition}

\begin{definition}[Locally testable codes; e.g., \citet{goldreich2006locally}]
For any $q \in \bbN$ and $\kappa > 0$,
an error-correcting code $\Enc \colon \{0,1\}^n \to \{0,1\}^\ell$ is said to be
\emph{$(q, \kappa)$-locally testable} if 
there exists a probabilistic polynomial-time algorithm $M$ that,
given oracle access to a string $f \in \{0,1\}^\ell$,
makes at most $q$ nonadaptive queries of $f$ and satisfies the following conditions:
\begin{itemize}
    \item (Completeness) If $f \in \Enc(\cdot)$, then 
    $M$ always accepts; namely,
    $
        \Pr[M^f \text{ accepts}] = 1.
    $
    \item (Soundness) If $f \notin \Enc(\cdot)$, then
    $M$ rejects with probability at least $\kappa \cdot \Delta(f, \Enc(\cdot))$; namely,
    $
        \Pr[M^f \text{ rejects}] \geq \kappa \cdot \Delta(f, \Enc(\cdot)).
    $
\end{itemize}
Such an algorithm $M$ is called a \emph{$(q, \kappa)$-local tester} for $\Enc$.
\end{definition}


\begin{theorem}[\cite{bensasson2006robust,bensasson2003randomness}]
\label{thm:LTC}
There exist $\rho, \kappa > 0$ and $q \in \bbN$ such that
for infinitely many $n$'s,
there exists a polynomial-time construction of 
a $(q,\kappa)$-locally testable error-correcting code
$\Enc_n \colon \{0,1\}^n \to \{0,1\}^{\ell(n)}$
with relative distance $\rho$ and $\ell(n) = n^{1+o(1)}$.
Moreover,
if the code $\Enc_n \colon \{0,1\}^n \to \{0,1\}^{\ell(n)}$ of the desired property exists for integer $n \in \bbN$,
then the next integer $n' \in \bbN$ for which
$\Enc_{n'} \colon \{0,1\}^{n'} \to \{0,1\}^{\ell(n')}$ exists is at most $n^{1+o(1)}$.
\end{theorem}

\subsection{Probabilistically Checkable Proofs of Proximity}

We formally define the notion of \emph{verifier}.
\begin{definition}[Verifier]
A \emph{verifier} with
\emph{randomness complexity} $r \colon \bbN \to \bbN$ and
\emph{query complexity} $q \colon \bbN \to \bbN$
is a probabilistic polynomial-time algorithm $V$ that 
given an input $x \in \{0,1\}^*$,
tosses $r = r(|x|)$ random bits $R$ and use $R$ to
generate a sequence of $q = q(|x|)$ queries
$I = (i_1, \ldots, i_q)$ and
a circuit
$D \colon \{0,1\}^q \to \{0,1\}$.
We write $(I,D) \sim V(x)$ to denote
the random variable for a pair of the query sequence and circuit generated by $V$ on input $x \in \{0,1\}^*$.
Denote by $V^{\pi}(x) \coloneq D(\pi|_I)$
the output of $V$ on input $x$ given oracle access to a proof $\pi \in \{0,1\}^*$.
We say that $V(x)$ \emph{accepts} a proof $\pi$
if $V^{\pi}(x) = 1$; i.e., $D(\pi|_I) = 1$ for $(I,D) \sim V(x)$.
\end{definition}

We proceed to the definition of \emph{PCPs of proximity} \cite{bensasson2006robust}
(a.k.a.~\emph{assignment testers} \cite{dinur2006assignment}).
For any pair language $L \subseteq \{0,1\}^* \times \{0,1\}^*$,
we define $L(x) \coloneq \{y \in \{0,1\}^* \mid (x,y) \in L\}$ for an input $x \in \{0,1\}^*$.

\begin{definition}[PCP of proximity \cite{bensasson2006robust,dinur2006assignment}]
A \emph{PCP of proximity} (PCPP) verifier for a pair language $L \subseteq \{0,1\}^* \times \{0,1\}^*$
with
\emph{proximity parameter} $\delta \in (0,1)$ and
\emph{soundness error} $s \in (0,1)$
is a verifier $V$ such that
for every pair of an explicit input $x \in \{0,1\}^*$ and
an implicit input oracle $y \in \{0,1\}^*$,
the following conditions hold:
\begin{itemize}
    \item (Completeness)
    If $(x,y) \in L$, there exists a proof $\pi \in \{0,1\}^*$ such that
    $V(x)$ accepts $y \circ \pi$ with probability $1$; namely,
    \begin{align}
        \exists \pi \in \{0,1\}^*, \;\; \Pr_{(I,D) \sim V(x)}\Bigl[ D((y \circ \pi)|_I) = 1 \Bigr] = 1.
    \end{align}
    \item (Soundness)
    If $y$ is $\delta$-far from $L(x)$,
    for every alleged proof $\pi \in \{0,1\}^*$,
    $V(x)$ accepts $y \circ \pi$ with probability less than $s$; namely,
    \begin{align}
        \forall \pi \in \{0,1\}^*, \;\; \Pr_{(I,D) \sim V(x)}\Bigl[ D((y \circ \pi)|_I) = 1 \Bigr] < s.
    \end{align}
\end{itemize}
\end{definition}

Subsequently, we introduce smooth PCPP verifiers.
We say that a verifier is \emph{smooth}
if each position in its proof is equally likely to be queried.

\begin{definition}[Smoothness]
A PCPP verifier $V$ is \emph{smooth} if
$V$ queries each position in
implicit input $y$ and proof $\pi$
with equal probability; namely,
there exists $p \in (0,1)$ such that
\begin{align}
p = \Pr_{(I,D) \sim V(x)}\Bigl[i \in I\Bigr] 
\end{align}
for every position $i$ of $y \circ \pi$.
\end{definition}

\begin{theorem}[Smooth PCPP \cite{paradise2021smooth,bensasson2006robust}]
\label{thm:smooth-PCPP}
For every pair language $L$ in $\NP$ and every number $\delta \in (0,1)$,
there exists a smooth PCPP verifier $V$ for $L$ with
randomness complexity $r(n) = \bigO(\log n)$,
query complexity $q(n) = \bigO(1)$,
proximity parameter $\delta$, and 
soundness error $s = 1-\Omega(\delta)$.
Moreover, for every pair $(x,y) \in L$,
a proof $\pi \in \{0,1\}^{\poly(n)}$ such that $V(x)$ always accepts $y \circ \pi$ can be 
constructed in polynomial time.
\end{theorem}

\subsection{Constraint Satisfaction Problems}
Here, we review constraint satisfaction problems.
\begin{definition}
    A \emph{$q$-ary constraint system} over variable set $N$ and alphabet $\Sigma$ is 
    defined as a collection of $q$-ary \emph{constraints},
    $\Psi = (\psi_j)_{j \in [m]}$,
    where each constraint $\psi_j \colon \Sigma^N \to \{0,1\}$
    depends on $q$ variables of $N$; namely,
    there exist $i_1, \ldots, i_q \in N$ and $f \colon \Sigma^q \to \{0,1\}$ such that
    $\psi_j(\asgmt) = f(\asgmt(i_1), \ldots, \asgmt(i_q))$
    for every $\asgmt \colon N \to \Sigma$.
\end{definition}
\noindent
For an assignment $\asgmt \colon N \to \Sigma$,
we say that $\asgmt$ \emph{satisfies} constraint $\psi_j$ if
$\psi_j(\asgmt) = 1$, and
$\asgmt$ \emph{satisfies} $\Psi$ if it satisfies all constraints of $\Psi$.
Moreover, we say that $\Psi$ is \emph{satisfiable} if $\Psi$ is satisfied by some assignment.
For an assignment $\asgmt \colon N \to \Sigma$,
its \emph{value} is defined as the fraction of constraints of $\Psi$ satisfied by $\asgmt$; namely,
\begin{align}
    \val_\Psi(\asgmt) \coloneq \frac{1}{|\Psi|} \cdot
    \left|\Bigl\{\psi_j \in \Psi \Bigm| \asgmt \text{ satisfies } \psi_j \Bigr\}\right|.
\end{align}

We refer to the equivalence between a PCP system and \prb{Gap $q$-CSP}
(see, e.g., \cite[Section~11.3]{arora2009computational}),
whose proof is included for the sake of completeness.

\begin{proposition}
\label{prp:PCP-CSP}
Let $V$ be a verifier with
randomness complexity $\bigO(\log n)$,
query complexity $\bigO(1)$, and
alphabet $\Sigma$,
and let $x \in \{0,1\}^*$ be an input.
Then, one can construct in polynomial time
a constraint system $\Psi = (\psi_j)_{j \in [m]}$
over $\poly(|x|)$ variables and alphabet $\Sigma$
such that
$\val_\Psi(\pi) = \Pr[V^\pi(x) = 1]$
for every proof $\pi \in \Sigma^{\poly(|x|)}$.

On the other hand,
for a $q$-ary constraint system $\Psi$
over variable set $N$ and alphabet $\Sigma$,
one can construct in polynomial time a verifier $V$ with
randomness complexity $\bigO(\log n)$,
query complexity $\bigO(1)$, and
alphabet $\Sigma$
such that
$\Pr[V^\asgmt = 1] = \val_\Psi(\asgmt)$
for every assignment $\asgmt \colon N \to \Sigma$.
\end{proposition}
\begin{proof}
Let $V$ be a verifier with
randomness complexity $r(n) = \bigO(n)$,
query complexity $q(n) = q \in \bbN$, and 
alphabet $\Sigma$.
Given an input $x \in \{0,1\}^*$,
we can assume the proof length for $V$ to be $\poly(|x|)$.
We construct a $q$-ary constraint system $\Psi$ 
over variable set $N \coloneq [\poly(|x|)]$ and alphabet $\Sigma$
as follows:
\begin{itemize}
\item for every possible sequence $R \in \{0,1\}^{r(|x|)}$ of $r(|x|)$ random bits,
we run $V(x)$ to generate
a query sequence $I_R = (i_1, \ldots, i_q)$ and 
a circuit $D_R \colon \Sigma^q \to \{0,1\}$ in polynomial time.
\item create a new constraint $\psi_j$ such that
\begin{align}
    \psi_j(\asgmt) \coloneq D_R(\asgmt(i_1), \ldots, \asgmt(i_q))
\end{align}
for every assignment $\asgmt \colon N \to \Sigma$.
\end{itemize}
Note that the construction of $\Psi$ completes in polynomial time;
in particular, the size of $\Psi$ is polynomial in $|x|$.
Observe that for any proof $\pi \in \Sigma^N$,
which can be thought of as an assignment to $\Psi$,
\begin{align}
    \val_\Psi(\pi) = \Pr_{\psi_j \sim \Psi}\Bigl[\psi_j(\pi) = 1\Bigr] 
    = \Pr_{(I,D) \sim V(S)}\Bigl[D(\pi|_I) = 1\Bigr] 
    = \Pr\Bigl[V^\pi(S) = 1\Bigr],
\end{align}
completing the proof of the first statement.
The second statement is omitted as can be shown similarly.
\end{proof}

\section{Probabilistic Checkable Reconfiguration Proofs}
\label{sec:PCRP}

In this section, we prove the main result of this paper, i.e.,
a PCRP verifier for a $\PSPACE$-complete reconfiguration problem.
For any pair of proofs $\pi^\sss, \pi^\ttt \in \Sigma^n$,
a \emph{reconfiguration sequence from $\pi^\sss$ to $\pi^\ttt$ over $\Sigma^n$}
is a sequence
$(\pi^{(1)}, \ldots, \pi^{(T)}) \in (\Sigma^n)^*$ such that
$\pi^{(1)} = \pi^\sss$,
$\pi^{(T)} = \pi^\ttt$, and 
$\pi^{(t)}$ and $\pi^{(t+1)}$ differ in at most one symbol
for every $t \in [T-1]$.

\begin{theorem}[Probabilistic Checkable Reconfiguration Proof (PCRP)]
\label{thm:reconfPCP}
A language $L$ is in $\PSPACE$ if and only if
there exists a verifier $V$
with randomness complexity $r(n)=\bigO(\log n)$ and
query complexity $q(n) = \bigO(1)$,
coupled with
polynomial-time computable proofs $\pi^\sss, \pi^\ttt \colon \allowbreak \{0,1\}^* \to \{0,1\}^*$
such that the following hold
for every input $x \in \{0,1\}^*$:
\begin{itemize}
    \item \textup{(Completeness)} If $x \in L$,
    there exists a reconfiguration sequence $\pi = ( \pi^{(1)}, \ldots \pi^{(T)} )$
    from $\pi^\sss(x)$ to $\pi^\ttt(x)$ over $\{0,1\}^{\poly(n)}$ such that
    $V(x)$ accepts every proof with probability $1$; namely,
    \begin{align}
        \forall t \in [T], \;\; \Pr\Bigl[ V^{\pi^{(t)}}(x) = 1 \Bigr] = 1.
    \end{align}
    \item \textup{(Soundness)} If $x \notin L$,
    every reconfiguration sequence $\pi = ( \pi^{(1)}, \ldots \pi^{(T)} )$
    from $\pi^\sss(x)$ to $\pi^\ttt(x)$ over $\{0,1\}^{\poly(n)}$
    includes a proof that is rejected by
    $V(x)$ with probability more than $\frac{1}{2}$; namely,
    \begin{align}
        \exists t \in [T], \;\; \Pr\Bigl[ V^{\pi^{(t)}}(x) = 1 \Bigr] < \frac{1}{2}.
    \end{align}
\end{itemize}
\end{theorem}

\subsection{$\PSPACE$-completeness of \prb{Succinct Graph Reachability}}
We first introduce a canonical $\PSPACE$-complete problem called \prb{Succinct Graph Reachability},
for which we design a PCRP system.

\begin{problem}
\label{prb:ckt-reach}
For a polynomial-size circuit
    $S \colon \{0,1\}^n \to \{0,1\}^n$ promised that
    $S(1^n) = 1^n$,
    \prb{Succinct Graph Reachability} requests to decide
    if 
    there is a sequence of a finite number of assignments,
    $( \alpha^{(1)}, \ldots, \alpha^{(T)} )$,
    from $0^n$ to $1^n$ such that
    $\alpha^{(t)} = \alpha^{(t+1)}$,
    $S(\alpha^{(t)}) = \alpha^{(t+1)}$, or
    $S(\alpha^{(t+1)}) = \alpha^{(t)}$ for all $t \in [T-1]$.
\end{problem}

Similar variants were formulated previously, e.g., \cite{galperin1983succinct,papadimitriou1986note}.
$\PSPACE$-completeness of \prb{Succinct Graph Reachability} is shown below
for the sake of completeness.

\begin{proposition}
\label{prp:ckt-reach}
\prb{Succinct Graph Reachability} is $\PSPACE$-complete.
\end{proposition}
\begin{proof} 
For the sake of convenience,
we first show that
a circuit $S \colon \{0,1\}^n \to \{0,1\}^n$
is a YES instance of \prb{Succinct Graph Reachability} 
if and only if
\begin{align}
    \exists m \in \bbN \text{ such that } \underbrace{S \circ \cdots \circ S}_{m \text{ times}} (0^n) = 1^n.
\end{align}
The ``if'' direction is obvious because
whenever $\underbrace{S \circ \cdots \circ S}_{m \text{ times}}(0^n) = 1^n$,
the sequence $( \alpha^{(1)}, \ldots, \alpha^{(m+1)} )$
such that $\alpha^{(t)} = \underbrace{S \circ \cdots \circ S}_{t-1 \text{ times}}(0^n)$ for all $t \in [m+1]$
satisfies the desired property.
Suppose we have a sequence $( \alpha^{(1)}, \ldots, \alpha^{(T)} )$
from $0^n$ to $1^n$ such that
$\alpha^{(t)} = \alpha^{(t+1)}$,
$S(\alpha^{(t)}) = \alpha^{(t+1)}$, or
$S(\alpha^{(t+1)}) = \alpha^{(t)}$.
It is not hard to see that for each $t$,
$\underbrace{S \circ \cdots \circ S}_{m_t \text{ times}}(\alpha^{(t)}) = 1^n$ for some $m_t \in \bbN$;
in particular, this is the case for $\alpha^{(1)} = 0^n$, completing the ``only if'' direction.

We now show the $\PSPACE$-completeness of \prb{Succinct Graph Reachability}.
Membership in $\PSPACE$ follows from
the fact that $\prb{Succinct Graph Reachability} \in \NPSPACE$ and
{Savitch}'s theorem \cite{savitch1970relationships}.
Consider the following $\PSPACE$-complete problem:
Given a deterministic Turing machine $M$,
input $x \in \{0,1\}^*$,
and $1^n$,
does $M$ accept $x$ in space $n$?
Note that all possible configurations of $M$ having space $n$ on input $x$ are specified by
$\{0,1\}^{\alpha n}$ for some constant $\alpha$ depending on $M$.
Let $c_\mathrm{init} \in \{0,1\}^{\alpha n}$ denote
the initial configuration of $M$ on input $x$, and
$M(x,c) \in \{0,1\}^{\alpha n}$ denote
the next configuration of $M$ following $c \in \{0,1\}^{\alpha n}$.
Define now a circuit
$S \colon \{0,1\}^{2+\alpha n} \to \{0,1\}^{2+\alpha n}$
of polynomial size (in $|M|$, $|x|$, and $n$) as follows:
\begin{align}
\begin{aligned}
    S(00 \circ 0^{\alpha n}) & \coloneq 01 \circ c_\mathrm{init}, \\
    S(01 \circ c) & \coloneq
    \begin{cases}
        11 \circ 1^{\alpha n} & \text{if } c \text{ is any accepting configuration}, \\
        00 \circ 0^{\alpha n} & \text{if } c \text{ is any rejecting configuration}, \\
        01 \circ M(x,c) & \text{otherwise},
    \end{cases}
    \\
    S(11 \circ 1^{\alpha n}) & \coloneq 11 \circ 1^{\alpha n}.
\end{aligned}
\end{align}
Observe easily that
$\underbrace{S \circ \cdots \circ S}_{m \text{ times}}(0^{2+\alpha n}) = 1^{2+\alpha n}$ 
for some $m \in \bbN$
if and only if
$M$ accepts $x$ in space $n$, completing the proof.
\end{proof}

\subsection{PCRP System for \prb{Succinct Graph Reachability}}
\label{subsec:main:verifier}
Here, we will construct a PCRP system for \prb{Succinct Graph Reachability}.
We first encode the assignment to a circuit by an error-correcting code.
For a polynomial-size circuit $S \colon \{0,1\}^n \to \{0,1\}^n$,
let $\Enc \colon \{0,1\}^n \to \{0,1\}^{\ell(n)}$ be
an $(\bigO(1), \kappa)$-locally testable error-correcting code with relative distance $\rho$
such that $\rho \in (0,1)$, $\kappa \in \bbN$, and $\ell(n) = n^{1+o(1)}$ by \cref{thm:LTC}.\footnote{
Without loss of generality, we can assume \cref{thm:LTC} holds for given $n$
because we can find $n' = n^{1+o(1)}$ for which \cref{thm:LTC} holds and
construct a slightly larger circuit $S' \colon \{0,1\}^{n'} \to \{0,1\}^{n'}$ such that
$S'$ is a YES instance if and only if so is $S$.
}
Let $M$ be an $(\bigO(1), \kappa)$-local tester for $\Enc$.
Consider the following pair language $L_\ckt \subseteq \{0,1\}^* \times \{0,1\}^*$:
Given a polynomial-size circuit $S \colon \{0,1\}^n \to \{0,1\}^n$ and
two strings $f, g \in \{0,1\}^{\ell(n)}$,
we define $(S, f \circ g) \in L_\ckt$ if and only if
$f = \Enc(\alpha)$ and $g = \Enc(\beta)$ for
a pair $\alpha, \beta \in \{0,1\}^n$ such that
$\alpha=\beta$, $S(\alpha) = \beta$, or $S(\beta) = \alpha$.
Intuitively, $L_\ckt$ determines the adjacency relation of codewords of $\Enc$ with respect to $S$.
Observe easily that $L_\ckt$ is in $\NP$, and
by \cref{prp:ckt-reach},
$S$ is
a YES instance of \prb{Succinct Graph Reachability}
if and only if 
there exists a sequence of strings,
$( f^{(1)}, \ldots, f^{(T)} )$,
from $\Enc(0^n)$ to $\Enc(1^n)$ such that
$(S, f^{(t)} \circ f^{(t+1)}) \in L_\ckt$
for all $t \in [T-1]$.
We will use $L_\ckt(S)$ to denote
the set of all strings $f \circ g \in \{0,1\}^{2\ell(n)}$
such that $(S, f \circ g) \in L_\ckt$; namely,
\begin{align}
\begin{aligned}
    L_\ckt(S)
    & \coloneq \Bigl\{
        f \circ g \in \{0,1\}^{2\ell(n)} \Bigm|
        (S, f \circ g) \in L_\ckt
    \Bigr\} \\
    & = \Bigl\{
        \Enc(\alpha) \circ \Enc(\beta) \Bigm|
        \alpha = \beta \vee S(\alpha)=\beta \vee S(\beta)=\alpha
    \Bigr\}.
\end{aligned}
\end{align}
Let $V_\ckt$ denote a smooth PCPP verifier for $L_\ckt$,
having
randomness complexity $r(n) = \bigO(\log n)$,
query complexity $q(n) = q \in \bbN$,
proximity parameter $\delta_\ckt \coloneq \frac{\rho}{4} \in (0,1)$, and
soundness error $s_\ckt \coloneq 1 - \Omega(\delta_\ckt) \in (0,1)$ obtained by \cref{thm:smooth-PCPP}.
Note that the proof length can be bounded by some polynomial $p(n)$ in the input length $n$.
Hereafter, we use a new symbol ``$\bot$'' that is neither $0$ nor $1$.
We will write $f \circ g \circ \pi$ for
a string provided to $V_\ckt(S)$, where
$f \circ g \in \{0,1,\bot\}^{\ell(n)} \times \{0,1,\bot\}^{\ell(n)}$ and
$\pi \in \{0,1,\bot\}^{p(n)}$ is an alleged proof.

\subsubsection{Verifier Description}
\label{subsubsec:main:verifier:description}
Our verifier $V$ is
given a polynomial-size circuit $S \colon \{0,1\}^n \to \{0,1\}^n$ and
oracle access to
$f \circ g \circ \pi \in \{0,1,\bot\}^{2\ell(n)+p(n)}$, and
is designed as follows:
\begin{enumerate}
\item $V$ ensures that $f$ \emph{or} $g$ must be a codeword of $\Enc$ by
running the local tester $M$ on $f$ and $g$ separately.
Note that $M$ rejects whenever it reads $\bot$ at least once,
which still ensures that $\Pr[M^f \text{ rejects}] \geq \kappa \cdot \Delta(f, \Enc(\cdot))$.
\item
$V$ allows $f \circ g$ to contain $\bot$,
enabling $f$ or $g$ to transform between different codewords of $\Enc$.
Specifically,
$V$ accepts with probability equal to the fraction of $\bot$ in $f$ or $g$,
which can be done by testing whether
$f_i = \bot$ or $g_j = \bot$ for
independently and uniformly chosen $i, j \in [\ell(n)]$.
During $f = \bot^n$ or $g = \bot^n$,
the contents of $\pi$ can be modified arbitrarily without being rejected,
which is essential in the perfect completeness (\cref{lem:completeness}).

\item On the other hand, if neither $f$ nor $g$ contains ``many'' $\bot$'s,
$V$ expects $f \circ g$ to be close to $L_\ckt(S)$; thus,
it wants to execute the smooth PCPP verifier $V_\ckt(S)$,
whose behavior is, however, undefined if $f \circ g \circ \pi$ contains $\bot$.
Instead, we run a \emph{modified verifier} $V'_\ckt(S)$,
which accepts if and only if 
$(f \circ g)|_I$ contains $\bot$ or
($\pi|_I$ does not contain $\bot$ and $D((f \circ g \circ \pi)|_I) = 1$)
for $(I,D) \sim V_\ckt(S)$.\footnote{
In the latter case, we can safely assume that $(f \circ g)|_I$ does not contain $\bot$.
}
This test is crucial for proving the soundness (\cref{lem:soundness}).
\end{enumerate}

The precise pseudocode of $V(S)$ is presented below.

\begin{itembox}[l]{
    Verifier $V^{f \circ g \circ \pi}(S)$ using
    local tester $M$ for $\Enc$ and
    smooth PCPP verifier $V_\ckt$ for $L_\ckt$.
}
    \begin{algorithmic}[1]
        \item[\textbf{Input:}]
            a polynomial-size circuit $S \colon \{0,1\}^n \to \{0,1\}^n$.
        \item[\textbf{Oracle access:}]
            strings $f,g \in \{0,1,\bot\}^{\ell(n)}$ representing an implicit input and
            a proof $\pi \in \{0,1,\bot\}^{p(n)}$.
        \State run local tester $M$ on $f$ and $g$.     \Comment{$M$ declares \textsf{reject} if it reads $\bot$.}
        \If{both runs of $M$ declare \textsf{reject}}
            \State \textsf{reject}.
        \EndIf
        \State pick $i \sim [\ell(n)]$ and $j \sim [\ell(n)]$ independently and uniformly.
        \If{$f_i \neq \bot$ and $g_j \neq \bot$}
            \LComment{run a modified PCPP verifier $V'_\ckt(S)$.}
            \State execute PCPP verifier $V_\ckt(S)$ to generate
            a query sequence $I = (i_1, \ldots, i_q)$ and 
            a circuit $D \colon \{0,1\}^q \to \{0,1\}$.
            \If{$(f \circ g)|_I$ contains $\bot$}
                \State \textsf{accept}.
            \ElsIf{$\pi|_I$ does not contain $\bot$ \textbf{and} $D((f \circ g \circ \pi)|_I) = 1$}
                \State \textsf{accept}.
            \Else
                \State \textsf{reject}.
            \EndIf
        \Else
            \State \textsf{accept}.
        \EndIf
    \end{algorithmic}
\end{itembox}

For any two strings
$\alpha, \beta \in \{0,1\}^n$ such that $\Enc(\alpha) \circ \Enc(\beta) \in L_\ckt(S)$,
let $\Pi(\alpha,\beta) \in \{0,1\}^{p(n)}$
denote a polynomial-time computable proof that makes
$V_\ckt(S)$ to accept
$\Enc(\alpha) \circ \Enc(\beta) \circ \Pi(\alpha, \beta)$ with probability $1$.
Note that
$\Enc(0^n) \circ \Enc(0^n) \circ \Pi(0^n, 0^n)$ and
$\Enc(1^n) \circ \Enc(1^n) \circ \Pi(1^n, 1^n)$
are accepted by $V(S)$ with probability $1$.

\subsubsection{Completeness and Soundness}
We now prove the completeness and soundness.
Define
$\sigma^\sss \coloneq \Enc(0^n) \circ \Enc(0^n) \circ \Pi(0^n, 0^n) \in \{0,1\}^{2\ell(n)+p(n)}$ and
$\sigma^\ttt \coloneq \Enc(1^n) \circ \Enc(1^n) \circ \Pi(1^n, 1^n) \in \{0,1\}^{2\ell(n)+p(n)}$.
Let $\val_V(S)$ denote
the maximum possible value of
\begin{align}
    \min_{\sigma^{(t)} \in \sqsigma} \Pr\Bigl[V(S) \text{ accepts } \sigma^{(t)} \Bigr]
\end{align}
over all possible reconfiguration sequences
$\sqsigma = (\sigma^{(1)}, \ldots, \sigma^{(T)})$
from
$\sigma^\sss$ to $\sigma^\ttt$.
The perfect completeness ensures that $\val_V(S) = 1$ if $S$ is a YES instance, while
the soundness guarantees that $\val_V(S) < 1-\delta$ for some $\delta \in (0,1)$ if $S$ is a NO instance.

We first show the completeness.

\begin{lemma}
\label{lem:completeness}
Suppose a circuit $S \colon \{0,1\}^n \to \{0,1\}^n$ is a YES instance of \prb{Succinct Graph Reachability}.
Then, there exists a reconfiguration sequence $\sqsigma$ from
$\sigma^\sss$ to $\sigma^\ttt$ over $\{0,1,\bot\}^{2 \ell(n) + p(n)}$ such that
$V(S)$ accepts any proof in $\sqsigma$ with probability $1$.
\end{lemma}
\begin{proof} 
It suffices to show that
for any $\alpha \neq \beta \in \{0,1\}^n$ such that
$\alpha = S(\beta)$ or $\beta = S(\alpha)$
(i.e., $\Enc(\alpha) \circ \Enc(\beta) \in L_\ckt(S)$),
there is a reconfiguration sequence $\sqsigma$ from
$\Enc(\alpha) \circ \Enc(\alpha) \circ \Pi(\alpha,\alpha)$
to $\Enc(\alpha) \circ \Enc(\beta) \circ \Pi(\alpha,\beta)$ such that
$V(S)$ accepts any proof in $\sqsigma$ with probability $1$.
Such a reconfiguration sequence is obtained by the following procedure:

\begin{itembox}[l]{Reconfiguration $\sqsigma$ from
$\Enc(\alpha) \circ \Enc(\alpha) \circ \Pi(\alpha,\alpha)$ to
$\Enc(\alpha) \circ \Enc(\beta) \circ \Pi(\alpha,\beta)$.}
\begin{algorithmic}[1]
\LComment{start from $\Enc(\alpha) \circ \Enc(\alpha) \circ \Pi(\alpha,\alpha)$.}
\State change the second string from $\Enc(\alpha)$ to $\bot^{\ell(n)}$ one by one.
\label{linum:complete:1}
\LComment{obtain $\Enc(\alpha) \circ \bot^{\ell(n)} \circ \Pi(\alpha,\alpha)$.}
\State change the proof from $\Pi(\alpha, \alpha)$ to $\Pi(\alpha,\beta)$ one by one.
\label{linum:complete:2}
\LComment{obtain $\Enc(\alpha) \circ \bot^{\ell(n)} \circ \Pi(\alpha,\beta)$.}
\State change the second string from $\bot^{\ell(n)}$ to $\Enc(\beta)$.
\label{linum:complete:3}
\LComment{end at $\Enc(\alpha) \circ \Enc(\beta) \circ \pi^{(\alpha,\beta)}$.}
\end{algorithmic}
\end{itembox}

By the following case analysis,  
$V(S)$ turns out to accept every intermediate proof $f \circ g \circ \pi$ with probability $1$, as desired.
\begin{itemize}
    \item (Line~\ref{linum:complete:1})
    $f \circ g \circ \pi$ is obtained from
    $\Enc(\alpha) \circ \Enc(\alpha) \circ \Pi(\alpha, \alpha)$ by replacing some symbols of the second $\Enc(\alpha)$ by $\bot$.
    Observe that the local tester $M$ always accepts $f = \Enc(\alpha)$.
    We show that $V'_\ckt(S)$ always accepts $f \circ g \circ \pi$.
    Let $(I,D) \sim V_\ckt(S)$.
    If $(f \circ g)|_I$ contains $\bot$, $V'_\ckt(S)$ accepts.
    Otherwise, since $\pi = \Pi(\alpha, \alpha)$ does not contain $\bot$,
    it holds that $(f \circ g \circ \pi)|_I = (\Enc(\alpha) \circ \Enc(\alpha) \circ \Pi(\alpha, \alpha))|_I$,
    implying $D((f \circ g \circ \pi)|_I) = 1$.

    \item (Line~\ref{linum:complete:2})
    $f \circ g \circ \pi$ has a form of
    $\Enc(\alpha) \circ \bot^{\ell(n)} \circ \pi$ for some $\pi \in \{0,1,\bot\}^{p(n)}$.
    The local tester $M$ always accepts $\Enc(\alpha)$, and
    $V(S)$ would not have run the modified verifier $V'_\ckt(S)$; i.e.,
    $V(S)$ always accepts $f \circ g \circ \pi$.

    \item (Line~\ref{linum:complete:3})
    $f \circ g \circ \pi$ is obtained from
    $\Enc(\alpha) \circ \Enc(\beta) \circ \Pi(\alpha, \beta)$
    by replacing some symbols of $\Enc(\beta)$ by $\bot$.
    Similarly to the first case, we can show that
    $V(S)$ always accepts $f \circ g \circ \pi$.
\end{itemize}

\end{proof}

We then show the soundness.

\begin{lemma}
\label{lem:soundness}
Suppose a circuit $S \colon \{0,1\}^n \to \{0,1\}^n$ is a NO instance of \prb{Succinct Graph Reachability}.
Then, for any reconfiguration sequence $\sqsigma$ from
$\sigma^\sss$ to $\sigma^\ttt$ 
over $\{0,1,\bot\}^{2 \ell(n) + p(n)}$,
$\sqsigma$ includes a proof that is rejected by $V(S)$ with probability at least
\begin{align}
    \min\left\{
        (\kappa \epsilon)^2, (1-\epsilon)^2 \cdot \frac{1-s_\ckt}{2}
    \right\},
    \text{ where }
    \epsilon \coloneq \min\left\{\frac{1-s_\ckt}{2q}, \frac{\rho}{3}\right\}.
\end{align}
\end{lemma}

By \cref{lem:completeness,lem:soundness},
we can complete the proof of \cref{thm:reconfPCP}.
\begin{proof}[Proof of \cref{thm:reconfPCP}]
We first prove the ``only if'' direction.
Since \prb{Succinct Graph Reachability} is $\PSPACE$-complete,
it is sufficient to create its verifier $V$ and
polynomial-time computable proofs $\pi^\sss$ and $\pi^\ttt$.
The verifier $V$ is described in \cref{subsubsec:main:verifier:description}.
For a polynomial-size circuit $S \colon \{0,1\}^n \to \{0,1\}$,
the number of queries that $V$ makes is bounded by
$2 \cdot (\text{\# queries of } M)
+ 2
+ (\text{\# queries of } V_\ckt)
= \bigO(1)$, and
the number of random bits that $V$ uses is bounded by
$2 \cdot (\text{random bits of } M)
+ 2 \cdot (\log \ell(n))
+ (\text{random bits of } V_\ckt)
= \bigO(\log n)$.
We define
$\pi^\sss \coloneq \Enc(0^n) \circ \Enc(0^n) \circ \Pi(0^n,0^n)$ and
$\pi^\ttt \coloneq \Enc(1^n) \circ \Enc(1^n) \circ \Pi(1^n,1^n)$,
which are polynomial-time computable.

We reduce the alphabet size of $V$ from three (i.e., $\{0,1,\bot\}$) to two.
Using \cref{prp:PCP-CSP}, we first convert $V(S)$ into
a constraint system $\Psi = (\psi_j)_{j \in [m]}$ over alphabet $\{0,1,\bot\}$
such that
$\Pr[V(S) \text{ accepts } \pi]$ is equal to $\val_\Psi(\pi)$
for any proof $\pi \in \{0,1,\bot\}^{\poly(n)}$.
By \cite{ohsaka2023gap},
we obtain
a constraint system $\Psi' = (\psi'_j)_{j \in [m']}$ over alphabet $\{0,1\}$ and 
its two satisfying assignments $A^\sss$ and $A^\ttt$ such that
$\val_\Psi(\pi^\sss \reco \pi^\ttt) = 1$ implies $\val_{\Psi'}(A^\sss \reco A^\ttt) = 1$, and
$\val_\Psi(\pi^\sss \reco \pi^\ttt) < 1-\epsilon$ implies
$\val_{\Psi'}(A^\sss \reco A^\ttt) < 1-\Omega(\epsilon)$.
Using \cref{prp:PCP-CSP} again,
we convert $\Psi'$ into a verifier $V'$
with randomness complexity $\bigO(\log n)$, query complexity $\bigO(1)$, and alphabet $\{0,1\}$ such that
$\Pr[V' \text{ accepts } \pi']$ is equal to $\val_{\Psi'}(\pi')$
for any proof $\pi' \in \{0,1\}^{\poly(n)}$.
Consequently,
if $S$ is a YES instance, by \cref{lem:completeness},
there exists a reconfiguration sequence $\sqasgmt$ from $\asgmt^\sss$ to $\asgmt^\ttt$ over $\{0,1\}^{\poly(n)}$
such that $V'$ accepts any proof in $\sqasgmt$ with probability $1$,
whereas if $S$ is a NO instance, by \cref{lem:soundness},
for any reconfiguration sequence $\sqasgmt$ from $\asgmt^\sss$ to $\asgmt^\ttt$ over $\{0,1\}^{\poly(n)}$,
$\sqasgmt$ includes a proof that is rejected by $V'$ with probability $\Omega(1)$,
which can be amplified to $\frac{1}{2}$ by a constant number of repetition, as desired.

We then prove the ``if'' direction.
Suppose a language $L$ admits
a verifier $V$ with
randomness complexity $r(n) = \bigO(\log n)$ and
query complexity $q(n) = \bigO(1)$, associated with
polynomial-time computable proofs $\pi^\sss$ and $\pi^\ttt$.
Consider then the following nondeterministic algorithm for finding a reconfiguration sequence from $\pi^\sss$ to $\pi^\ttt$.
\begin{itembox}[l]{Nondeterministic polynomial-space algorithm for finding a reconfiguration sequence.}
\begin{algorithmic}[1]
    \item[\textbf{Input:}] $x \in \{0,1\}^*$.
    \State compute proofs $\pi^\sss(x)$ and $\pi^\ttt(x)$ that are accepted by $V$ with probability $1$.
    \State let $\pi^{(0)} \coloneq \pi^\sss(x)$ \textbf{and} $t \leftarrow 0$.
    \Repeat
        \If{$\pi^{(t)} = \pi^\ttt(x)$}
            \State \textsf{accept}.
        \EndIf
        \State nondeterministically guess the next proof $\pi^{(t+1)} \in \{0,1\}^{\poly(|x|)}$.
        \State check if
            $\pi^{(t)}$ and $\pi^{(t+1)}$ differ in at most one bit, and
            $V$ accepts $\pi^{(t+1)}$ with probability $1$
            by enumerating all possible $r(|x|)$ random bits.
        \If{the above test passes}
            \State forget $\pi^{(t)}$ \textbf{and} let $t \leftarrow t+1$.
        \Else
            \State \textsf{reject}.
        \EndIf
    \Until{$t > 2^{\poly(|x|)}$}
    \State \textsf{reject}.
\end{algorithmic}
\end{itembox}
The above algorithm accepts $x$ if and only if $x \in L$.
Moreover, it requires polynomial space and terminates within a finite steps;
namely, $L \in \NPSPACE$.
By {Savitch}'s theorem \cite{savitch1970relationships}, $L \in \PSPACE$.
\end{proof}

The remainder of this section is devoted to the proof of \cref{lem:soundness}.

\paragraph{Proof of \cref{lem:soundness}.}
Suppose we are given a reconfiguration sequence
$\sqsigma = (\sigma^{(1)}, \ldots, \sigma^{(T)})
= ( f^{(1)} \circ g^{(1)} \circ \pi^{(1)}, \ldots, f^{(T)} \circ g^{(T)} \circ \pi^{(T)} )$
from $\sigma^\sss$ to $\sigma^\ttt$
such that $\val_V(S) = \min_{t \in [T]}\Pr[V(S) \text{ accepts } \sigma^{(t)}]$.
Define
\begin{align}
\label{eq:soundness:epsilon}
    \epsilon \coloneq \min\left\{\frac{1-s_\ckt}{2q}, \frac{\rho}{3}\right\},
\end{align}
where
$q$ is the query complexity of $V_\ckt$,
$s_\ckt$ is the soundness error of $V_\ckt$, and
$\rho$ is the relative distance of $\Enc$.
Observe that if \emph{both} $f^{(t)}$ and $g^{(t)}$ for some $t \in [T]$ are $\epsilon$-far from $\Enc(\cdot)$,
then
$M$ rejects each $f^{(t)}$ and $g^{(t)}$ with probability more than $\kappa \epsilon$; namely,
\begin{align}
\Pr\Bigl[V(S) \text{ rejects } \sigma^{(t)} \Bigr]
\geq \Pr\Bigl[M \text{ rejects } f\Bigr] \cdot \Pr\Bigl[M \text{ rejects } g\Bigr]
> (\kappa\epsilon)^2.
\end{align}
Hereafter, we assume that
$f^{(t)}$ or $g^{(t)}$ is $\epsilon$-close to $\Enc(\cdot)$ for every $t \in [T]$.

We then define
$\Dec \colon \{0,1\}^{\ell(n)} \to \{0,1\}^n \cup \{\ast\}$ as
\begin{align}
    \Dec(f) \coloneq
    \begin{dcases}
        \argmin_{\alpha \in \{0,1\}^n} \Delta(f, \Enc(\alpha))
        & \text{if } f \text{ is } \epsilon\text{-close to } \Enc(\cdot), \\
        \ast & \text{otherwise},
    \end{dcases}
\end{align}
where $\ast$ means ``undefined''.
Using $\Dec$,
we obtain a sequence from $(0^n, 0^n)$ to $(1^n, 1^n)$,
denoted $( (\alpha^{(1)}, \beta^{(1)}), \ldots, (\alpha^{(T)}, \beta^{(T)}) )$,
where
$\alpha^{(t)} \coloneq \Dec(f^{(t)})$ and
$\beta^{(t)} \coloneq \Dec(g^{(t)})$ for all $t \in [T]$.
By assumption,
$\alpha^{(t)}$ or $\beta^{(t)}$ must not be $\ast$ for all $t \in [T]$.
We claim the following (see also \cref{tab:gamma}):

\begin{table}[t]
    \newcommand{\pick}{\cellcolor{red!40}}
    \centering
    \small
    \tabcolsep = 2pt
    \begin{tabular}{|c||c|c|c|c|c|c|c|c|c|c|c|c|c|c|c|c|c|c|c|}
        \hline
        $t$ &
        $1$ & $\cdots$ & $t_1$ & $t_1+1$ & 
        $\cdots$ & $t_2-1$ & $t_2$ & $t_2+1$ &
        $\cdots$ & $t_3-1$ & $t_3$ & $t_3+1$ &
        $\cdots$ & $t_4-1$ & $t_4$ & $t_4+1$ &
        $\cdots$ & $T-1$ & $T$ \\
        \hline\hline
        $\alpha^{(t)}$ &
            \pick$0^n$ & \pick$\cdots$ & \pick$0^n$ & $\ast$ &
            $\cdots$ & $\ast$ & \pick$\alpha^{(t_2)}$ & \pick$\alpha^{(t_2)}$ &
            \pick$\cdots$ & \pick$\alpha^{(t_2)}$ & \pick$\alpha^{(t_2)}$ & $\ast$ &
            $\cdots$ & $\ast$ & \pick$1^n$ & \pick$1^n$ &
            \pick$\cdots$ & \pick$1^n$ & \pick$1^n$ \\
        \hline
        $\beta^{(t)}$ &
            $0^n$ & $\cdots$ & \pick$0^n$ & \pick$0^n$ &
            \pick$\cdots$ & \pick$0^n$ & \pick$0^n$ & $\ast$ &
            $\cdots$ & $\ast$ & \pick$\beta^{(t_3)}$ & \pick$\beta^{(t_3)}$ &
            \pick$\cdots$ & \pick$\beta^{(t_3)}$ & \pick$\beta^{(t_3)}$ & $\ast$ &
            $\cdots$ & $\ast$ & $1^n$ \\ 
        \hline
    \end{tabular}
    \caption{
    Illustration of \cref{clm:soundness:gamma},
    which finds a sequence $\sqgamma$ from $0^n$ to $1^n$ over $\{0,1\}^n$
    using $( (\alpha^{(1)}, \beta^{(1)}), \ldots, (\alpha^{(T)}, \beta^{(T)}) )$.
    Colored strings are included in $\sqgamma$, resulting in
    $\sqgamma = ( 0^n, \ldots, 0^n, \alpha^{(t_2)}, \ldots, \alpha^{(t_2)}, \beta^{(t_3)}, \ldots, \beta^{(t_3)}, 1^n, \ldots, 1^n )$.
    If an input circuit $S$ is a NO instance,
    at least one of
    $\Enc(\alpha^{(t_2)}) \circ \Enc(0^n)$,
    $\Enc(\alpha^{(t_2)}) \circ \Enc(\beta^{(t_3)})$, or
    $\Enc(\beta^{(t_3)}) \circ \Enc(1^n)$ is not in $L_\ckt(S)$.
    }
    \label{tab:gamma}
\end{table}

\begin{claim}
\label{clm:sound:alpha-beta}
The following hold:
\begin{enumerate}[label=\textup{(P\arabic*)}]
    \item $\alpha^{(t)} = \alpha^{(t+1)}$ or $\beta^{(t)} = \beta^{(t+1)}$ for each $t$.
    \label{clm:sound:alpha-beta:1}
    \item If $\alpha^{(t)} \neq \ast$ and $\alpha^{(t+1)} \neq \ast$, then $\alpha^{(t)} = \alpha^{(t+1)}$.
    \label{clm:sound:alpha-beta:2}
    \item If $\beta^{(t)} \neq \ast$ and $\beta^{(t+1)} \neq \ast$, then $\beta^{(t)} = \beta^{(t+1)}$.
    \label{clm:sound:alpha-beta:3}
\end{enumerate}    
\end{claim}
\begin{proof} 
Suppose first $\alpha^{(t)} \neq \alpha^{(t+1)}$ and $\beta^{(t)} \neq \beta^{(t+1)}$
for some $t$.
Then, $f^{(t)} \circ g^{(t)}$ and $f^{(t+1)} \circ g^{(t+1)}$
differ in at least two symbols,
contradicting that $\sqsigma$ is a reconfiguration sequence; thus, \ref{clm:sound:alpha-beta:1} must hold.

Suppose then
$\alpha^{(t)} \neq \ast$,
$\alpha^{(t+1)} \neq \ast$, and
$\alpha^{(t)} \neq \alpha^{(t+1)}$.
Since $f^{(t)}$ and $f^{(t+1)}$ are assumed to be $\epsilon$-close to $\Enc(\cdot)$,
by triangle inequality,
we have
\begin{align}
\begin{aligned}
    \Delta(f^{(t)}, f^{(t+1)})
    & \geq \underbrace{\Delta(\Enc(\alpha^{(t)}), \Enc(\alpha^{(t+1)}))}_{\geq \rho}
        - \underbrace{\Delta(f^{(t)}, \Enc(\alpha^{(t)}))}_{\leq \epsilon}
        - \underbrace{\Delta(f^{(t+1)}, \Enc(\alpha^{(t+1)}))}_{\leq \epsilon} \\
    & \geq \rho - 2\epsilon,
\end{aligned}
\end{align}
implying that $f^{(t)}$ and $f^{(t+1)}$
differ  in $(\rho-2\epsilon) \cdot \ell(n) \geq 2$ bits (for sufficiently large $n$),
contradicting that $\sqsigma$ is a reconfiguration sequence; thus, \ref{clm:sound:alpha-beta:2} holds.
Similarly, \ref{clm:sound:alpha-beta:3} can be shown.
\end{proof}

Now, we can find a \emph{valid} sequence from $0^n$ to $1^n$ over $\{0,1\}^n$,
denoted $\sqgamma = ( \gamma^{(1)}, \ldots, \gamma^{(T')} )$,
along a ``path'' over a grid $\{\alpha, \beta\} \times [T]$
by the following procedure (see also \cref{tab:gamma}):
\begin{itembox}[l]{Sequence $\sqgamma$ from $0^n$ to $1^n$.}
\begin{algorithmic}
    \State let $\alpha^{(T+1)} \coloneq 1^n$ \textbf{and} $\beta^{(T+1)} \coloneq 1^n$ for convenience.
    \State let $t' \leftarrow 0$
        \textbf{and}
        place a token at $(\alpha, 1)$.
    \Repeat
        \State let $t' \leftarrow t'+1$.
        \If{token is at $(\alpha, t)$}
            \State let $\gamma^{(t')} \coloneq \alpha^{(t)}$.
            \If{$\alpha^{(t+1)} \neq \ast$}
                token goes to $(\alpha, t+1)$.
            \Else
                \hspace{24mm} token goes to $(\beta, t)$.
            \EndIf
        \EndIf
        \If{token is at $(\beta, t)$}
            \State let $\gamma^{(t')} \coloneq \beta^{(t)}$.
            \If{$\beta^{(t+1)} \neq \ast$}
                token goes to $(\beta, t+1)$.
            \Else
                \hspace{24mm} token goes to $(\alpha, t)$.
            \EndIf
        \EndIf
    \Until{token is at $(\alpha, T+1)$ or $(\beta, T+1)$}
    \State \textbf{return} $\sqgamma \coloneq (\gamma^{(1)}, \ldots, \gamma^{(t')})$.
\end{algorithmic}
\end{itembox}

The correctness of the above procedure is shown in the following claim.
\begin{claim}
\label{clm:soundness:gamma}
The above procedure terminates and returns a sequence
$\sqgamma = ( \gamma^{(1)}, \ldots, \gamma^{(T')} )$
from $0^n$ to $1^n$ consisting only of strings of $\{0,1\}^n$.
Moreover,
each pair $(\gamma^{(t')}, \gamma^{(t'+1)})$ is equal to either
$(\alpha^{(t)}, \alpha^{(t+1)})$,
$(\beta^{(t)}, \beta^{(t+1)})$,
$(\alpha^{(t)}, \beta^{(t)})$, or
$(\beta^{(t)}, \alpha^{(t)})$ for some $t$.
\end{claim}
\begin{proof} 
\ref{clm:sound:alpha-beta:1} of \cref{clm:sound:alpha-beta} ensures that
\begin{itemize}
    \item if $\alpha^{(t)} \neq \ast$, then $\alpha^{(t+1)} \neq \ast$ or $\beta^{(t)} \neq \ast$;\footnote{
    Because otherwise, we have
    $\alpha^{(t)} \neq \ast$ and $\alpha^{(t+1)} = \beta^{(t)} = \ast$,
    which implies $\beta^{(t+1)} = \ast$ by \ref{clm:sound:alpha-beta:1},
    a contradiction that $\alpha^{(t+1)}$ or $\beta^{(t+1)}$ is not $\ast$.
    }
    \item if $\beta^{(t)} \neq \ast$, then $\beta^{(t+1)} \neq \ast$ or $\alpha^{(t)} \neq \ast$.
\end{itemize}
Thus, by construction, $\sqgamma$ does not include any $\ast$.

Suppose the token is currently placed at $(\alpha, t)$ and just before at $(\beta, t)$.
Then, the token must be placed at $(\alpha, t+1)$ in the next step.
Similarly, if the token is placed at $(\beta, t)$ and just before at $(\alpha, t)$; then,
it must be at $(\beta, t+1)$ in the next step,
which ensures that we eventually reach $(\alpha, T+1)$ or $(\beta, T+1)$ to terminate.
The latter statement is obvious from the construction.
\end{proof}

Since we have been given a NO instance,
$\sqgamma$ must include
$(\gamma^{(t')}, \gamma^{(t'+1)})$ for some $t' \in [T']$ such that
$\gamma^{(t')} \neq \gamma^{(t'+1)}$,
$S(\gamma^{(t')}) \neq \gamma^{(t'+1)}$, and
$S(\gamma^{(t'+1)}) \neq \gamma^{(t')}$.
By \ref{clm:sound:alpha-beta:2} and \ref{clm:sound:alpha-beta:3} 
and \cref{clm:soundness:gamma},
either of
$(\alpha^{(t)}, \beta^{(t)}) = (\gamma^{(t')}, \gamma^{(t'+1)})$ or
$(\beta^{(t)}, \alpha^{(t)}) = (\gamma^{(t')}, \gamma^{(t'+1)})$ must hold for some $t \in [T]$,
implying that
$\alpha^{(t)} \neq \beta^{(t)}$, 
$S(\alpha^{(t)}) \neq \beta^{(t)}$, and
$S(\beta^{(t)}) \neq \alpha^{(t)}$;
namely,
$\Enc(\alpha^{(t)}) \circ \Enc(\beta^{(t)}) \notin L_\ckt(S)$.

We now estimate the probability that $V(S)$ rejects
$\sigma^{(t)}$.
Recall that $f^{(t)}$ and $g^{(t)}$ are $\epsilon$-close to $\Enc(\cdot)$, implying that
\begin{align}
    \Pr_{\substack{i,j \sim [n]}}\Bigl[
        f^{(t)}_i \neq \bot \text{ and } g^{(t)}_j \neq \bot
    \Bigr] \geq (1-\epsilon)^2.
\end{align}
So, we would run the modified PCPP verifier $V'_\ckt$ with probability $\geq (1-\epsilon)^2$.
We use the following claim to bound the rejection probability of $V'_\ckt(S)$:

\begin{claim}
\label{clm:soundness:PCPP}
Suppose $\alpha = \Dec(f) \in \{0,1\}^n$, $\beta = \Dec(g) \in \{0,1\}^n$, and
$\Enc(\alpha) \circ \Enc(\beta) \notin L_\ckt(S)$ for
$f \circ g \in \{0,1,\bot\}^{2\ell(n)}$.
Then, for every proof $\pi \in \{0,1,\bot\}^{p(n)}$,
the modified PCPP verifier $V'_\ckt(S)$ rejects $f \circ g \circ \pi$ with probability more than
$1-s_\ckt - \epsilon q$.
\end{claim}
\begin{proof} 
We first show that
$(f \circ g)|_I$ contains $\bot$ for $(I,D) \sim V_\ckt(S)$ with probability at most $\epsilon q$.
Denote by $I_\ckt$ the indices of $f \circ g \circ \pi$, where $|I_\ckt| = 2\ell(n) + p(n)$.
By smoothness of $V_\ckt$, we have
$p_\ckt \coloneq \Pr_{(I,D)}[i \in I]$ for all $i \in I_\ckt$.
Since $|I| = q$ for any $(I,D) \sim V_\ckt(S)$, we obtain
\begin{align}
    |I_\ckt| \cdot p_\ckt
    = \sum_{i \in I_\ckt} \Pr_{(I,D)}\Bigl[i \in I\Bigr]
    \leq q
    \implies p_\ckt \leq \frac{q}{2\ell(n) + p(n)}.
\end{align}
Using a union bound and
the assumption that each $f$ and $g$ contains $\bot$ in at most $\epsilon \cdot \ell(n)$ positions,
we derive
\begin{align}
\begin{aligned}
    \Pr_{(I,D)}\Bigl[(f \circ g)_{I} \text{ contains } \bot\Bigr]
    & = \Pr_{(I,D)}\Bigl[\exists i \in I \text{ s.t. } (f \circ g)_i = \bot\Bigr] \\
    & \leq \sum_{i : (f\circ g)_i = \bot}
        \Pr_{(I,D)}\Bigl[i \in I\Bigr] \\
    & \leq \epsilon \cdot 2 \ell(n) \cdot p_\ckt \leq \epsilon q.
\end{aligned}
\end{align}

Subsequently, we show that $f \circ g$ is
$\delta_\ckt$-far from $L_\ckt(S)$, where
$\delta_\ckt = \frac{\rho}{4}$.
Letting $(\alpha^\star, \beta^\star) \in \{0,1\}^n \times \{0,1\}^n$ such that
$\Enc(\alpha^\star) \circ \Enc(\beta^\star) \in L_\ckt(S)$,
we have $\alpha \neq \alpha^\star$ or $\beta \neq \beta^\star$.
Suppose $\alpha \neq \alpha^\star$; then, we have
\begin{align}
\Delta(f, \Enc(\alpha^\star))
\geq \Delta(\Enc(\alpha^\star), \Enc(\alpha)) - \Delta(\Enc(\alpha), f)
\geq \rho - \epsilon.
\end{align}
Similarly, $\Delta(g, \Enc(\beta^\star)) \geq \rho-\epsilon$ if $\beta \neq \beta^\star$.
Consequently, we obtain
\begin{align}
    \Delta(f \circ g, \Enc(\alpha^\star) \circ \Enc(\beta^\star))
    \geq \frac{\rho-\epsilon}{2}
    > \frac{\rho}{4} = \delta_\ckt,
\end{align}
where we used the fact that $\epsilon \leq \frac{\rho}{3}$.

Taking a union bound, we derive
\begin{align}
\begin{aligned}
    & \Pr\Bigl[\text{modified verifier } V'_\ckt(S) \text{ accepts } f \circ g \circ \pi\Bigr] \\
    & = \Pr_{(I,D)}\Bigl[(f \circ g)|_I \text{ contains } \bot \text{ or }
        \Bigl( \pi|_I \text{ doesn't contain } \bot \text{ and } D((f\circ g \circ \pi)|_I) = 1 \Bigr)\Bigr] \\
    & = \Pr_{(I,D)}\Bigl[(f \circ g)|_I \text{ contains } \bot \text{ or }
        \Bigl( (f \circ g \circ \pi)|_I \text{ doesn't contain } \bot \text{ and } D((f\circ g \circ \pi)|_I) = 1 \Bigr)\Bigr] \\
    & \leq \underbrace{\Pr_{(I,D)}\Bigl[(f \circ g)|_I \text{ contains } \bot\Bigr]}_{\leq \epsilon q}
    + \Pr_{(I,D)}\Bigl[(f \circ g \circ \pi)|_I \text{ doesn't contain } \bot \text{ and } D((f\circ g \circ \pi)|_I) = 1\Bigr]. \\
\end{aligned}
\end{align}
Let $\tilde{\pi}$ be a proof obtained from $\pi$ by replacing every occurrence of $\bot$ by $0$.
If $(f \circ g \circ \pi)|_I$ does not contain $\bot$ and $D((f \circ g \circ \pi)|_I) = 1$, then 
$D((f \circ g \circ \tilde{\pi})|_I) = 1$.
Since
$f \circ g$ is $\delta_\ckt$-far from $L_\ckt(S)$,
we have
\begin{align}
\begin{aligned}
    & \Pr_{(I,D)}\Bigl[(f \circ g \circ \pi)|_I \text{ doesn't contain } \bot \text{ and } D((f\circ g \circ \pi)|_I) = 1\Bigr] \\
    & \leq \Pr_{(I,D)}\Bigl[ D((f\circ g \circ \tilde{\pi})|_I) = 1 \Bigr] \\
    & = \Pr\Bigl[ V_\ckt(S) \text{ accepts } f\circ g \circ \tilde{\pi} \Bigr] \\
    & < s_\ckt.
\end{aligned}
\end{align}
Accordingly, we get
\begin{align}
\begin{aligned}
    \Pr\Bigl[V'_\ckt(S) \text{ rejects } f \circ g \circ \pi\Bigr] 
    = 1 - \Pr\Bigl[V'_\ckt(S) \text{ accepts } f \circ g \circ \pi\Bigr] 
    > 1 - s_\ckt - \epsilon q,
\end{aligned}
\end{align}
completing the proof.
\end{proof}

Using \cref{clm:soundness:PCPP} and the definition of $\epsilon$ in \cref{eq:soundness:epsilon},
we derive
\begin{align}
\begin{aligned}
    \Pr\Bigl[V(S) \text{ rejects } \sigma^{(t)} \Bigr]
    & \geq \Pr_{i,j \sim [\ell(n)]}\Bigl[f^{(t)}_i \neq \bot \text{ and } g^{(t)}_j \neq \bot\Bigr]
    \cdot \Pr\Bigl[V'_\ckt(S) \text{ rejects } \sigma^{(t)} \Bigr] \\
    & > (1-\epsilon)^2 \cdot (1 - s_\ckt - \epsilon q) \\
    & \underbrace{\geq}_{\epsilon \leq \frac{1-s_\ckt}{2q}} (1-\epsilon)^2 \cdot \frac{1-s_\ckt}{2},
\end{aligned}
\end{align}
Consequently, we get
\begin{align}
    \max_{t \in [T]} \Pr\Bigl[V(S) \text{ rejects } \sigma^{(t)} \Bigr]
    > \min\left\{
        (\kappa \epsilon)^2, (1-\epsilon)^2 \cdot \frac{1-s_\ckt}{2}
    \right\},
\end{align}
accomplishing the proof of \cref{lem:soundness}. \qed

\subsection{Impossibility of Extension to Average Case}
Since the verifier of \cref{thm:reconfPCP}
co-nondeterministically guesses $t \in [T]$ and probabilistically checks $\pi^{(t)}$,
one might think of extending it so as to choose $t \in [T]$ randomly.
The soundness case then requires that
the verifier accepts ``most'' of but rejects a constant fraction of the proofs in any reconfiguration sequence.
The resulting reconfiguration proof $(\pi^{(1)}, \ldots, \pi^{(T)})$
can be thought of as a (kind of) \emph{rectangular PCPs} \cite{bhangale2020rigid},
whose column is of exponential length and row is of polynomial length, and
we pick $t \in [T]$ and run a verifier on $\pi^{(t)}$ to decide whether to accept.
However, such relaxation is impossible, as formally stated below.

\begin{observation}
\label{obs:impossible}
Let $V$ be a verifier with randomness complexity $r(n) = \bigO(\log n)$ and
query complexity $q(n) = \bigO(1)$,
$x \in \{0,1\}^n$ be an input of length $n$, and
$\pi^\sss(x), \pi^\ttt(x) \in \{0,1\}^{\poly(n)}$
be a pair of proofs that are accepted by $V$ with probability $1$.
Then, there always exists a reconfiguration sequence
$\sqpi = ( \pi^{(1)}, \ldots, \pi^{(T)} )$
from $\pi^\sss(x)$ to $\pi^\ttt(x)$ over $\{0,1\}^{\poly(n)}$ such that
\begin{align}
    \Pr_{t \sim [T]}\Bigl[V \text{ accepts } \pi^{(t)}\Bigr]
    = \Pr_{\substack{t \sim [T] \\ (I,D) \sim V(x)}}\Bigl[D(\pi^{(t)}|_I) = 1 \Bigr]
    > 1-\frac{1}{2^{\Omega(n)}},
\end{align}
where the probability is over the $r(n)$ random bits of $V$ and $t \sim [T]$.
\end{observation}
\begin{proof} 
Consider any reconfiguration sequence
$( \pi^{(1)}, \ldots \pi^{(\poly(n)+1)} )$ from $\pi^\sss(x)$ to $\pi^\ttt(x)$.
By appending $( \underbrace{\pi^\ttt(x), \ldots, \pi^\ttt(x)}_{2^n - (\poly(n)+1) \text{ times}} )$ to it,
we obtain an exponential-length reconfiguration sequence $\sqpi = ( \pi^{(1)}, \ldots, \pi^{(2^n)} )$.
Since $\sqpi$ contains (at least) $2^n - (\poly(n)+1)$ number of $\pi^\ttt(x)$, we have
\begin{align}
    \Pr_{t \sim [T]}\Bigl[V \text{ accepts } \pi^{(t)}\Bigr]
    \geq \frac{2^n - (\poly(n)+1)}{2^n} = 1 - \frac{1}{2^{\Omega(n)}},
\end{align}
as desired.
\end{proof}

\section{$\PSPACE$-hardness of Approximation for Reconfiguration Problems}
\label{sec:inapprox}

In this section, we show that many popular reconfiguration problems are 
$\PSPACE$-hard to approximate,
answering an open problem of \cite{ito2011complexity}.
Since \citet{ohsaka2023gap} has already shown gap-preserving reductions starting from
the Reconfiguration Inapproximability Hypothesis (RIH),
which asserts that a gap version of \prb{Maxmin CSP Reconfiguration} is $\PSPACE$-hard,
we prove that RIH is true.

\subsection{Constant-factor Inapproximability of \prb{Maxmin CSP Reconfiguration}}
\label{sec:inapprox:csp}
We first define reconfiguration problems on constraint satisfaction.
For a $q$-ary constraint system $\Psi = (\psi_j)_{j \in [m]}$ over variable set $N$ and alphabet $\Sigma$ and
its two satisfying assignments $\asgmt^\sss$ and $\asgmt^\ttt$ for $\Psi$,
a \emph{reconfiguration sequence from $\asgmt^\sss$ to $\asgmt^\ttt$ over $\Sigma^N$}
is a sequence $( \asgmt^{(1)}, \ldots, \asgmt^{(T)} ) \in (\Sigma^N)^*$
such that
$\asgmt^{(1)} = \asgmt^\sss$,
$\asgmt^{(T)} = \asgmt^\ttt$, and
$\asgmt^{(t)}$ and $\asgmt^{(t+1)}$ differ in at most one vertex for every $t \in [T-1]$.
In the \prb{$q$-CSP Reconfiguration} problem,
we are asked to decide if there is a reconfiguration sequence of
satisfying assignments to $\Psi$ from $\asgmt^\sss$ to $\asgmt^\ttt$.
Subsequently, we 
formulate an approximate variant of \prb{$q$-CSP Reconfiguration} \cite{ito2011complexity,ohsaka2023gap}.
For a reconfiguration sequence
$\sqasgmt = ( \asgmt^{(1)}, \ldots, \asgmt^{(T)} )$ of assignments,
let $\val_\Psi(\sqasgmt)$ denote the \emph{minimum fraction} of satisfied constraints
over all $\asgmt^{(i)}$'s in $\sqasgmt$; namely,
\begin{align}
    \val_\Psi(\sqasgmt) \coloneq \min_{\asgmt^{(i)} \in \sqasgmt} \val_\Psi(\asgmt^{(i)}).
\end{align}
In \prb{Maxmin $q$-CSP Reconfiguration},
we wish to maximize $\val_\Psi(\sqasgmt)$
subject to $\sqasgmt = ( \asgmt^\sss, \ldots, \asgmt^\ttt )$.
For two assignments $\asgmt^\sss, \asgmt^\ttt \colon N \to \Sigma$,
let $\val_\Psi(\asgmt^\sss \reco \asgmt^\ttt)$ denote the maximum value of $\val_\Psi(\sqasgmt)$
over all possible reconfiguration sequences $\sqasgmt$ from $\asgmt^\sss$ to $\asgmt^\ttt$; namely,
\begin{align}
\val_\Psi(\asgmt^\sss \reco \asgmt^\ttt)
\coloneq \max_{\sqasgmt = ( \asgmt^\sss, \ldots, \asgmt^\ttt )} \val_\Psi(\sqasgmt)
= \max_{\sqasgmt = ( \asgmt^\sss, \ldots, \asgmt^\ttt )} \min_{\asgmt^{(i)} \in \sqasgmt} \val_\Psi(\asgmt^{(i)}).
\end{align}
For every numbers $0 \leq s \leq c \leq 1$,
\prb{Gap$_{c,s}$ $q$-CSP Reconfiguration} requests to determine
for a $q$-ary constraint system $\Psi$ and
its two assignments $\asgmt^\sss$ and $\asgmt^\ttt$,
whether
$\val_\Psi(\asgmt^\sss \reco \asgmt^\ttt) \geq c$ or
$\val_\Psi(\asgmt^\sss \reco \asgmt^\ttt) < s$.
As a corollary of \cref{thm:reconfPCP,prp:PCP-CSP}, we immediately obtain a proof of RIH, as formally stated below.

\begin{theorem}
\label{thm:RIH}
The Reconfiguration Inapproximability Hypothesis holds; that is,
there exist a universal constant $q \in \bbN$ such that
\prb{Gap$_{1,\frac{1}{2}}$ $q$-CSP Reconfiguration} with alphabet size $2$ is $\PSPACE$-complete.
\end{theorem}

By \cref{thm:RIH} and \cite{ohsaka2023gap,ohsaka2024gap},
a host of reconfiguration problems turn out to be $\PSPACE$-hard to approximate within a constant factor,
as listed in \cref{cor:RIH}.

\subsection{Polynomial-factor Inapproximability of \prb{Maxmin Clique Reconfiguration}}

We amplify inapproximability of \prb{Maxmin Clique Reconfiguration}
from a constant factor to a polynomial factor.
The proof of the following result uses the \emph{derandomized graph product} due to
\citet{alon1995derandomized};
see also \cite[Section~3.3.2]{hoory2006expander} and \cite[Example~22.7]{arora2009computational}.

\begin{theorem}
\label{thm:clique}
There exists a constant $\epsilon \in (0,1)$ such that
\prb{Maxmin Clique Reconfiguration} is $\PSPACE$-hard to approximate
within a factor of $n^{\epsilon}$,
where $n$ is the number of vertices.
\end{theorem}

We here formulate \prb{Clique Reconfiguration} and its approximate variant.
Denote by $\omega(G)$ the clique number of a graph $G$.
For a pair of cliques $C^\sss$ and $C^\ttt$ of a graph $G$,
a \emph{reconfiguration sequence from $C^\sss$ to $C^\ttt$}
is a sequence $(C^{(1)}, \ldots, C^{(T)})$ of cliques of $G$ such that
$C^{(1)} = C^\sss$, 
$C^{(T)} = C^\ttt$, and
$C^{(t)}$ and $C^{(t+1)}$ differ in at most one vertex; i.e.,
$|C^{(t)} \Delta C^{(t+1)}| = 1$.\footnote{
Such a model of reconfiguration is called \emph{token addition and removal} \cite{ito2011complexity}.
}
\prb{Clique Reconfiguration} asks if
there is a reconfiguration sequence 
from $C^\sss$ to $C^\ttt$
made up of cliques only of size at least
$\min\{ |C^\sss|, |C^\ttt| \} - 1$.
For a reconfiguration sequence of cliques of $G$,
denoted $\scrC = ( C^{(1)}, \ldots, C^{(T)} )$,
let
\begin{align}
    \val_G(\scrC) \coloneq \min_{C^{(i)} \in \scrC} |C^{(i)}|.
\end{align}
Then, for a pair of cliques $C^\sss$ and $C^\ttt$ of $G$,
\prb{Maxmin Clique Reconfiguration} requires to maximize 
$\val_G(\scrC)$ subject to
$\scrC = ( C^\sss, \ldots, C^\ttt )$.
Subsequently,
let $\val_G(C^\sss \reco C^\ttt)$ denote
the maximum value of $\val_G(\scrC)$
over all possible reconfiguration sequences $\scrC$
from $C^\sss$ to $C^\ttt$; namely,
\begin{align}
\val_G(C^\sss \reco C^\ttt) \coloneq \max_{\scrC = ( C^\sss, \ldots, C^\ttt )} \val_G(\scrC).
\end{align}

\paragraph{Reduction.}
We first describe a gap-amplification reduction from
\prb{Maxmin Clique Reconfiguration} to itself
using the derandomized graph product \cite{alon1995derandomized}.
Let $(G, C^\sss, C^\ttt)$ be an instance of 
\prb{Maxmin Clique Reconfiguration},
where $G=(V,E)$ is a graph on $n$ vertices.
By \cref{thm:RIH} and \cite{ohsaka2023gap},
it is $\PSPACE$-hard to distinguish whether
$\val_G(C^\sss \reco C^\ttt) \geq \omega(G)-1$ or
$\val_G(C^\sss \reco C^\ttt) \geq (1-\epsilon)(\omega(G)-1)$
for some constant $\epsilon \in (0,1)$
even when
$|C^\sss| = |C^\ttt| = \omega(G)$ and
$\frac{\omega(G)}{n}\in \left[\frac{1}{3}, \frac{1}{2}\right]$.

Construct then a new instance $(H,D^\sss,D^\ttt)$
of \prb{Maxmin Clique Reconfiguration} as follows.
Let $\ell = \lceil \log n \rceil$, and
$X$ be a $(d,\lambda)$-expander graph over the same vertex as $G$.
The precise value of $d$ and $\lambda$ will be determined later.
Graph $H = (W,F)$ is defined as follows:
\begin{itemize}
    \item \textbf{Vertex set:} $W$ is 
    the set consisting of all length-$(\ell-1)$ walks 
    $\vec{w} = ( w_1, \ldots, w_\ell )$ over $X$.
    Note that the number of vertices is
    equal to $N \coloneq |W| = nd^{\ell-1}$, which is polynomial in $n$.
    \item \textbf{Edge set:} $H$ contains an edge
    between $\vec{w}_1 \neq \vec{w}_2 \in W$ if and only if
    a subgraph of $G$ induced by $\vec{w}_1 \cup \vec{w}_2$ forms a clique.
\end{itemize}
For any clique $C \subseteq V$ of $G$,
define $D_C \subseteq W$ as
\begin{align}
    D_C \coloneq \Bigl\{\vec{w} \in W \Bigm| \vec{w} \subseteq C \Bigr\},
\end{align}
which is a clique of $H$ as well.
Constructing $D^\sss \coloneq D_{C^\sss}$ and
$D^\ttt \coloneq D_{C^\ttt}$
completes the reduction.
We refer to the following property about random walks over expander graphs.
\begin{lemma}[\cite{alon1995derandomized}]
\label{lem:clique:random-walk}
    Let  $S$ be any vertex set of $X$, and
    $\vec{Z} \coloneq (Z_1, \ldots, Z_\ell)$ a $\ell$-tuple of random variables denoting 
    the vertices of a uniformly chosen $(\ell-1)$-length random walk over $X$.
    Then, it holds that
    \begin{align}
        \left(\frac{|S|}{|V|} - 2\frac{\lambda}{d}\right)^\ell
        \leq \Pr_{\vec{Z}}\Bigl[\forall i \in [\ell], \; Z_i \in S \Bigr]
        \leq \left(\frac{|S|}{|V|} + 2\frac{\lambda}{d}\right)^\ell.
    \end{align}
\end{lemma}

The completeness and soundness are shown below.
\begin{lemma}
\label{lem:clique:completeness}
If $\val_G(C^\sss \reco C^\ttt) \geq \omega(G) - 1$,
then
\begin{align}
    \val_H(D^\sss \reco D^\ttt)
    \geq |W| \cdot \left(\frac{\omega(G)-1}{|V|} - 2\frac{\lambda}{d}\right)^{\ell}.
\end{align}
\end{lemma}
\begin{proof} 
It suffices to consider the case that
$C^\sss$ and $C^\ttt$ differ in exactly two vertices
(i.e., $C^\ttt$ is obtained from $C^\sss$ by removing and adding a single vertex).
There is a reconfiguration sequence $( C^\sss, C^\circ, C^\ttt )$ from $C^\sss$ to $C^\ttt$,
where $C^\circ \coloneq C^\sss \cap C^\ttt$ is a clique of size $\omega(G)-1$.
Since $D_{C^\sss} \supset D_{C^\circ}$ and $D_{C^\ttt} \supset D_{C^\circ}$ by definition,
we can reconfigure from $D^\sss = D_{C^\sss}$ to
$D^\ttt = D_{C^\ttt}$ by
first removing the vertices of $D_{C^\sss} \setminus D_{C^\circ}$ one by one and 
then adding the vertices of $D_{C^\ttt} \setminus D_{C^\circ}$ one by one.
Thus, we have $\val_H(D^\sss \reco D^\ttt) \geq |D_{C^\circ}|$.
\cref{lem:clique:random-walk} derives that
\begin{align}
    \frac{|D_{C^\circ}|}{|W|}
    = \Pr_{\vec{Z}}\Bigl[\forall i \in [\ell], \; Z_i \in C^\circ \Bigr]
    \geq \left(\frac{|C^\circ|}{|V|} - 2\frac{\lambda}{d}\right)^\ell
    = \left(\frac{\omega(G)-1}{|V|} - 2\frac{\lambda}{d}\right)^{\ell},
\end{align}
completing the proof.
\end{proof}

\begin{lemma}
\label{lem:clique:soundness}
If $\val_G(C^\sss \reco C^\ttt) < (1-\epsilon)(\omega(G)-1)$, then
\begin{align}
    \val_H(D^\sss \reco D^\ttt) < |W| \cdot
        \left((1-\epsilon)\frac{\omega(G)-1}{|V|}  +2\frac{\lambda}{d}\right)^\ell.
\end{align}
\end{lemma}
\begin{proof} 
We show the contrapositive.
Suppose we are given
a reconfiguration sequence $\scrD = ( D^{(1)}, \ldots, D^{(T)} )$
from $D^\sss$ to $D^\ttt$ such that
$\val_H(\scrD) \geq |W| \cdot \left((1-\epsilon)\frac{\omega(G)-1}{|V|}  +2\frac{\lambda}{d}\right)^\ell$.
For any clique $D$ of $H$,
define $C_D$ as 
\begin{align}
    C_D \coloneq \bigcup_{\vec{w} \in D} \vec{w},
\end{align}
which is a clique of $G$ as well.
Observe that 
$\val_G(C_{D^{(t)}} \reco C_{D^{(t+1)}}) \geq \min\{|C_{D^{(t)}}|, |C_{D^{(t+1)}}|\}$
for any $t \in [T-1]$ since
$C_{D^{(t)}} \subset C_{D^{(t+1)}}$ or $C_{D^{(t)}} \supset C_{D^{(t+1)}}$,
implying further that
\begin{align}
\begin{aligned}
    \val_G(C^\sss \reco C^\ttt)
    & \geq \min_{t \in [T-1]} \val_G(C_{D^{(t)}} \reco C_{D^{(t+1)}}) \\
    & \geq \min_{t \in [T-1]} \min\Bigl\{|C_{D^{(t)}}|, |C_{D^{(t+1)}}|\Bigr\} \\
    & \geq \min_{D^{(t)} \in \scrD} |C_{D^{(t)}}|.
\end{aligned}
\end{align}
\cref{lem:clique:random-walk} derives that 
\begin{align}
    \frac{|D^{(t)}|}{|W|}
    = \Pr_{\vec{Z}}\Bigl[\vec{Z} \in D^{(t)}\Bigr]
    \leq \Pr_{\vec{Z}}\Bigl[\forall i \in [\ell], \; Z_i \in C_{D^{(t)}}\Bigr]
    \leq \left( \frac{|C_{D^{(t)}}|}{|V|} + 2\frac{\lambda}{d} \right)^\ell.
\end{align}
On the other hand, by assumption,
\begin{align}
    \frac{|D^{(t)}|}{|W|} \geq \left((1-\epsilon)\frac{\omega(G)-1}{|V|}  +2\frac{\lambda}{d}\right)^\ell.
\end{align}
Consequently, we have $|C_{D^{(t)}}| \geq (1-\epsilon)(\omega(G)-1)$ for all $t \in [T]$; thus,
$\val_G(C^\sss \reco C^\ttt) \geq (1-\epsilon)(\omega(G)-1)$, as desired.
\end{proof}

We are now ready to accomplish the proof of \cref{thm:clique}.

\begin{proof}[Proof of \cref{thm:clique}]
Letting $X$ satisfy $\frac{\lambda}{d} < \frac{\epsilon}{32}$ so that
$\frac{\lambda}{d} < \frac{1}{8}\frac{\omega(G)-1}{n} \epsilon$
for sufficiently large $n$, we have
\begin{align}
\label{eq:clique:lower}
    \frac{\omega(G)-1}{|V|} - 2\frac{\lambda}{d} &
    \geq \frac{\omega(G)-1}{|V|}\left(1-\frac{\epsilon}{4}\right) \\
\label{eq:clique:upper}
    (1-\epsilon)\frac{\omega(G)-1}{|V|} + 2\frac{\lambda}{d} &
    < \frac{\omega(G)-1}{|V|}\left(1-\frac{3}{4}\epsilon\right).
\end{align}
Such an expander graph $X$ can be constructed in polynomial time in $n$, e.g.,
by using an explicit construction of near-Ramanujan graphs
\cite{mohanty2021explicit,alon2021explicit}.
By \cref{lem:clique:completeness,lem:clique:soundness} and \cref{eq:clique:lower,eq:clique:upper},
it is $\PSPACE$-hard to approximate \prb{Maxmin Clique Reconfiguration} within a factor of
\begin{align}
    \frac{|W| \cdot \left( \frac{\omega(G)-1}{n} - 2\frac{\lambda}{d} \right)^\ell}{|W| \cdot \left( (1-\epsilon)\frac{\omega(G)-1}{n} + 2\frac{\lambda}{d} \right)^\ell} 
    > \nu^\ell,
    \text{ where }
    \nu \coloneq \frac{1-\frac{\epsilon}{4}}{1-\frac{3}{4}\epsilon}.
\end{align}
Suppose $\nu^\ell = N^\delta$ for some $\delta$; then, $\delta$ should be
\begin{align}
\begin{aligned}
    & \nu^{\lceil \log n \rceil} = (n \cdot d^{\lceil \log n \rceil -1})^\delta \\
    & \implies
    \lceil \log n \rceil \cdot \log \nu = \delta \cdot (\log n + (\lceil \log n \rceil - 1) \cdot \log d) \\
    & \implies
    \delta = \frac{\lceil \log n \rceil \cdot \log \nu }{\log n + (\lceil \log n \rceil - 1) \cdot \log d}
    = \Theta\left(\frac{\log \nu}{1 + \log d}\right).
\end{aligned}
\end{align}
Consequently,
\prb{Maxmin Clique Reconfiguration} is $\PSPACE$-hard to approximate within a factor of 
$N^\delta$ for some $\delta \in (0,1)$.
\end{proof}


\printbibliography

\end{document}